\documentclass[10pt,a4paper]{amsart}
\usepackage{amsmath,amsthm,amsfonts,amssymb}
\usepackage[colorlinks=true]{hyperref}

\theoremstyle{plain}
\newtheorem{theorem}{Theorem}[section]
\newtheorem{lemma}[theorem]{Lemma}

\newtheorem{proposition}[theorem]{Proposition}

\newtheorem{hyp}[theorem]{Assumption}
\newtheorem*{hyp+}{Stronger assumption}

\theoremstyle{definition}

\newtheorem{remark}[theorem]{Remark}
\newtheorem*{remark*}{Remark}
\newtheorem{example}[theorem]{Example}

\newcommand\dis
{\displaystyle}

\newcommand\C{{\mathbf C}}
\newcommand\R{{\mathbf R}}
\newcommand\N{{\mathbf N}}
\newcommand\Sch{{\mathcal S}}
\newcommand\F{\mathcal F}
\newcommand\M{\mathcal M}
\renewcommand\O{\mathcal O}
\newcommand{\LL}{\mathcal{L}}

\newcommand\virgp{\raise 2pt\hbox{,}}

\renewcommand\le{\leslant}
\renewcommand\ge{\geslant}

\def\virgp{\raise 2pt\hbox{,}}

\def\({\left(}
\def\){\right)}
\def\<{\left\langle}
\def\>{\right\rangle}
\def\le{\leqslant}
\def\ge{\geqslant}

\newcommand\Eq[2]{\mathop{\sim}\limits_{#1\rightarrow#2}}
\newcommand\Tend[2]{\mathop{\longrightarrow}\limits_{#1\rightarrow#2}}

\def\d{\partial}
\def\g{\gamma}
\def\eps{\varepsilon}
\def\l{\lambda}

\def\si{{\sigma}}
\newcommand{\tQ}{\widetilde{Q}}

\DeclareMathOperator{\vect}{span}

\numberwithin{equation}{section}

\DeclareMathOperator{\re}{Re}
\DeclareMathOperator{\im}{Im}

\begin{document}

\title[Minimal blow-up for inhomogeneous NLS]{Minimal blow-up
  solutions to the mass-critical inhomogeneous NLS equation} 
\author[V. Banica]{Valeria Banica}
\address[V. Banica]{D\'epartement de Math\'ematiques\\ Universit\'e
  d'Evry\\ Bd. F.~Mitterrand\\ 91025 Evry\\ France} 
\email{Valeria.Banica@univ-evry.fr}
\author[R. Carles]{R\'emi Carles}
\address[R. Carles]{Univ. Montpellier~2\\Math\'ematiques
\\CC~051\\F-34095 Montpellier}
\address{CNRS, UMR 5149\\  F-34095 Montpellier\\ France}
\email{Remi.Carles@math.cnrs.fr}
\author[T. Duyckaerts]{Thomas Duyckaerts}
\address[T. Duyckaerts]{D\'epartement de Math{\'e}matiques\\
  Universit{\'e} de Cergy-Pontoise\\CNRS UMR 8088\\ 
2 avenue Adolphe Chauvin\\ BP 222, Pontoise\\ 95302 Cergy-Pontoise
cedex\\ France}
\email{tduyckae@math.u-cergy.fr}
\begin{abstract}
  We consider the mass-critical focusing nonlinear Schr\"odinger
  equation in the presence of an external potential, when the
  nonlinearity is inhomogeneous. We show that if the inhomogeneous
  factor in front of the nonlinearity is sufficiently flat at a
  critical point, then there exists a
  solution  which blows up in finite time with the maximal
  (unstable) rate at this point. In the case where
    the critical point is a maximum, this solution has minimal mass
    among the blow-up solutions. As a corollary, we also obtain
    unstable blow-up solutions of the mass-critical Schr\"odinger
    equation on some surfaces. The proof is based 
  on properties of the linearized operator around the ground state,
  and on a full use of the invariances of the equation with an
  homogeneous nonlinearity and no potential, \emph{via} time-dependent
  modulations.  
\end{abstract}
\maketitle

\section{Introduction}
\label{sec:intro}

\subsection{Setting of the problem and main result}
We consider the equation
\begin{gather}
\label{NLS}
i\partial_t u +\Delta u -V(x)u+g(x)|u|^{4/d}u=0,\quad x\in \R^d,\\
\notag
u_{\mid t=0}=u_0\in H^{1}\(\R^d\)
\end{gather}
where $d=1$ or $d=2$, $g$ and $V$ are real smooth functions on $\R^d$,
bounded as well as 
their derivatives, and $g$ is positive at least in an open subset of
$\R^d$. We investigate blowing up solutions to \eqref{NLS}. One of the 
applications that we have in mind is the study of finite time blow-up
for solutions to the nonlinear Schr\"odinger equation on surfaces. The
link between these two problems is detailed in
\S\ref{sec:presentationsurf} below.
\smallbreak

First, let us recall some classical arguments (see
e.g. \cite{Ca03}). The nonlinearity is 
energy-subcritical, so for any initial condition $u_0\in H^1$, there
exists a maximal interval of existence $]T_-(u_0),T_+(u_0)[$, and a
solution $u$ of \eqref{NLS} such that 
$$ u\in C(]T_-,T_+[,H^1).$$
Furthermore if $T_+<+\infty$, then $\dis\lim_{t\rightarrow T_+} \|\nabla
u(t)\|_{L^2}=+\infty.$
The mass $M=\|u(t)\|^2_{L^2}$ and the energy 
$$ E=\int \(\frac{1}{2}|\nabla u(t,x)|^2+\frac{1}{2}
V(x)|u(t,x)|^2-\frac{1}{\frac{4}{d}+2} g(x)|u(t,x)|^{\frac{4}{d}+2}\)dx$$ 
are independent of $t\in ]T_-,T_+[$. 
\smallbreak

We consider the ground state $Q$, which
is (up to translations) the unique
positive solution of the equation   
$$ \Delta Q+Q^{1+4/d}=Q, \quad x\in \R^d.$$
Recall that $Q$ is $C^{\infty}$, radial, and exponentially decreasing
at infinity. Furthermore, $Q$ is the critical point for the
Gagliardo--Nirenberg inequality (\cite{We83})
\begin{equation}
\label{GN}
\|\psi\|_{L^{2+4/d}}^{2+4/d}\le C\,\|\nabla
\psi\|_{L^2}^2 \,\|\psi\|_{L^2}^{4/d},\qquad\forall \psi\in
H^1\left(\R^d\right). 
\end{equation}
In the homogeneous case $V=0$, $g=1$, the equation 
\begin{equation}
  \label{eq:nls}
  i\d_t u+\Delta u + \lvert u\rvert^{4/d}{u}=0
\end{equation}
is stable by the 
pseudo-conformal transformation: 
if $u$ is a solution of \eqref{eq:nls}, so is $\widetilde{v}(t,x)$ defined by
\begin{equation}
\label{pseudo_conformal}
\widetilde v(t,x)=(\mathcal{T} u)(t,x):=\frac{e^{i\frac{|x|^2}{4t}}}{t^{d/2}}
\overline{u}\left(\frac 1t,\frac xt\right). 
\end{equation} 
Applying this transformation to the stationary
solution $e^{it}Q$, one gets a solution of Equation \eqref{eq:nls}
$$S(t,x)=e^{-i/t}\,\frac{e^{i\frac{|x|^2}{4t}}}{t^{d/2}}\, 
Q\left(\frac{x}{t}\right),$$
that blows up at time $t=0$, and such that $\|\nabla
u(t)\|_{L^2}\approx \frac 1t$  
as $t\to 0$. A classical argument shows, as a consequence of
\eqref{GN} that this is  
the minimal mass solution blowing-up in finite time (see \cite{We83}). 
\smallbreak

When equation \eqref{eq:nls} is perturbed in such a way that the
pseudo-conformal transformation is no longer valid, there are only few known
examples of blow-up solutions with the same growth rate. 
Consider the same equation \eqref{eq:nls}
posed on an open subset of $\R^d$ (with Dirichlet or Neumann boundary
conditions) or on a flat torus. Then one can construct  
blow-up solutions as perturbation of $S(t,x)$ with an exponentially
small error when $t\to 0$: see \cite{OgTs91} for $d=1$ and
\cite{BuGeTz03} for $d=2$.  
The proof relies on a fixed point
argument around a truncation of $S(t,x)$. The linear term is
considered as a source term, and is controlled in spaces of functions
decaying exponentially in time. This approach was first used in
\cite{Me90} to construct solutions with several blow-up points. 
\smallbreak

The recent work \cite{KrLeRa09} is devoted to a $4$-dimensional 
mass critical Hartree equation, with an inhomogeneous kernel. In the
corresponding homogeneous case, when the Hartree term is
$(|x|^{-2}\ast |u|^2)u$,  
\eqref{pseudo_conformal} leaves the equation invariant, yielding a blow-up
solution analogous to $S(t,x)$.  
Under the assumption that the perturbation
  vanishes at some large order at the 
blow-up point, a pseudo-conformal, minimal mass blow-up solution of
the perturbed equation is constructed. In this
case, the solution is only a 
polynomial perturbation of the explicit ground state pseudo-conformal
blow-up solution, and the proof of \cite{OgTs91} and \cite{BuGeTz03}
is no longer valid. The construction of \cite{KrLeRa09} relies on an
adaptation of an argument of Bourgain and Wang \cite{BoWa97}. 
\smallbreak

In the general setting of
\eqref{NLS}, the strategy of
\cite{OgTs91} and \cite{BuGeTz03} does not work either unless both $g$ and $V$
are constant around the blow-up point. The argument of Bourgain and
Wang is easy to adapt and gives, as in \cite{KrLeRa09}, a  minimal
mass solution under strong flatness conditions on $g$ and $V$ at 
the blow-up point (see Remark~\ref{rem:easy} and
Section~\ref{sec:plain}). These flatness conditions and the
concentration of the solution at the blow-up point imply that the
terms induced by $g$ and $V$ are small at the blow-up time. Our goal
is to weaken as much as possible 
these conditions: we construct blow-up solutions for any bounded potential
$V$ with bounded derivatives, assuming only a 
vanishing condition to the order $2$ on $g-g(x_0)$ at the blow-up
point $x_0$. We assume for simplicity that $x_0=0$ and that $g(0)=1$,
the general case $x_0\in \R^d$, $g(x_0)>0$ follows by space translation
and scaling. 
For $s\ge 0$, we denote
$$ \Sigma^s=\left\{\psi\in H^s\left(\R^d\right)\;\Big|\; |x|^s \psi\in
  L^2\left(\R^d\right)\right\}=H^s\left(\R^d\right)\cap
\F\left(H^s\left(\R^d\right)\right),$$  
and we shall drop the index for $\Sigma^1$. Our main result  is the following:

\begin{theorem}\label{thblup}
Let $d=1$ or $d=2$ and $V\in C^2(\R^d;\R)$, $g\in
C^4(\R^d;\R)$. Assume that $\d^\beta V\in L^\infty$ for $|\beta|\le
2$, $\d^\alpha g\in L^\infty$ for $|\beta|\le 4$, and   
\begin{equation}
  \label{eq:hypg}
 g(0)=1\quad ;\quad 
\frac{\partial g}{\partial x_j}(0)=\frac{\partial^2 g}{\partial
   x_j\partial x_k}(0)=0,\; 1\le j,k\le d. 
\end{equation}
Then there exist $T>0$, $u \in C(]0,T[,\Sigma)$ solution of
\eqref{NLS} on $]0,T[$ such that 
\begin{equation}
\label{UQ}
 \left\lVert u(t)-\widetilde{S}(t)\right\rVert_\Sigma\underset{t\rightarrow 
 0^+}{\longrightarrow}0,\text{ with }\widetilde{S}(t,x)=
 e^{-itV(0)}\,\frac{e^{i\frac{|x|^2}{4t}-i\theta\(\frac{1}{t}\)}}{\l(t)^{d/2}}\, 
Q\(\frac{x-x(t)}{\l(t)}\) ,
\end{equation}
where $\theta,\l$ are continuous real-valued functions and $x(t)$ is a
continuous $\R^d$-valued function such 
that 
\begin{align*}
 \theta
(\tau)&=\tau +o(\tau) \quad \text{as } \tau\to +\infty,\\
\l(t)& \sim t\text{ and }|x(t)|=o(t)\quad \text{as }t\to 0^+. 
\end{align*}
\end{theorem}
\begin{remark}\label{rem:moreexplicit}
 We easily infer the asymptotics
 \begin{equation*}
    \left\lVert u(t)-\widetilde
      S_2(t)\right\rVert_{\F(H^1)}\underset{t\rightarrow  
 0^+}{\longrightarrow}0,\text{ with }\widetilde{S}_2(t,x)=
 e^{-itV(0)}\,\frac{e^{i\frac{|x|^2}{4t}-i\theta\(\frac{1}{t}\)}}{t^{d/2}}\, 
Q\(\frac{x}{t}\). 
 \end{equation*}
Note that in this formula, we do not control the $\dot H^1$-norm, for
which a better control of $\l$ and $x(t)$ would be needed.
\end{remark}
As explained below, we
construct the blow-up solution as a perturbation of the solution
$\widetilde S_2(t,x)$.  The flatness condition on $g$ implies that the
new perturbative 
terms induced by the inhomogeneity $g$ are small as $t$ tends to $0$. 
\begin{remark}
The pseudo-conformal blow-up regime of Theorem~\ref{thblup}, where the
blow-up rate $\|\nabla u(t)\|_{L^2}$ is of order $1/t$ around $t=0$,
is unstable and non-generic, as opposed to the blow-up regime at a
rate $\left(\frac{{\log(|\log |t||)}}{{t}}\right)^{1/2}$ 
highlighted (in space dimension $1$) by G.~Perelman \cite{Pe01} (see
also \cite{MeRa05}). This
log-log regime was shown to be generic in all dimensions, under a
spectral assumption if $d\ge 2$, in a series of papers of F.~Merle and
P.~Rapha\"el (see e.g 
\cite{MeRa04,Ra05}). This assumption was checked in the case $d\le 4$,
and the main properties of the log-log regime persist for $d=5$ (see
\cite{FiMeRa06}). Theorem~\ref{thblup} may also be seen as a structural
stability property for the pseudo-conformal blow-up regime: this
regime persists under some perturbations of the equation. 
\end{remark}

\begin{remark}
Note that \eqref{UQ} implies  $\|u(t)\|_2^2=\|Q\|_2^2$. If we
assume furthermore that $|g|\le 1$, the solution constructed in
Theorem~\ref{thblup} has minimal mass for blow-up. This is consistent
with the conjecture that the non-generic blow-up occurs at the
boundary of the manifold of all blowing-up solutions. Note also that
$g$ may not remain everywhere positive: we consider a localized phenomenon.  
\end{remark}
\begin{remark}
\label{rem:easy}
Establishing Theorem \ref{thblup} is much easier if we assume that
$V-V(0)$ and $g-g(0)$ vanish to high order at $x=0$. This is the
analogue of Theorem~1 of \cite{KrLeRa09} in the 
context of Hartree equation. In Section~\ref{sec:plain}, we give, in
this less general setting, a short proof of \eqref{UQ} which is an
adaptation of 
\cite{BoWa97} and simplifies the argument of \cite{KrLeRa09}. In this
case we can assume that 
$\theta(\tau)=\tau$, $\l(t)=t$ and $x(t)=0$. The first equality should
also hold (in 
view of the recent work \cite{RS-p}) in the general context of Theorem
\ref{thblup}. The main difficulty of the 
proof of Theorem \ref{thblup} under the general assumption is to
combine the strategy of \cite{BoWa97} with modulation theory to relax
the high order flatness assumption to the weaker assumption
\eqref{eq:hypg}. This difficulty already  appears in \cite{KrSc09} in
a more delicate context (see below). 
\end{remark}
We next discuss two particular cases. If $g\equiv 1$, our theorem
shows that for any real-valued smooth potential $V$ which is bounded
on $\R^d$ as well as all its derivatives, for any point $x_0\in \R^d$,
there exists a solution of  
\begin{equation*}
  i\partial_t u+\Delta u-V(x) u+|u|^{4/d}u=0
\end{equation*}
blowing-up at $x_0$ at a pseudo-conformal rate.
Little is known about blow-up solutions for this equation, except
in some particular cases (where $V$ is unbounded) where algebraic miracles
provide a good understanding: if 
$V$ is linear in $x$, Avron--Herbst formula shows that $V$ does not
change the blow-up rate (\cite{CaNa04}). If $V(x)=\pm \omega^2|x|^2$,
$V$ changes the blow-up time, but not the blow-up rate (\cite{CaSIMA}). 
Our result shows that the $S(t)$ blow-up rate remains for any
bounded potential (\emph{e.g.} obtained after truncating the above
potential). 

\smallbreak

Equation~\eqref{NLS} in the case $V\equiv 0$ was studied by F.~Merle
in \cite{Me96}.  
Assume for the sake of simplicity that
$$ g(0)=1 \quad \text{and}\quad \forall x\neq 0,\ |g(x)|<1.$$
In this case, $g$ attains its maximum at $0$.
In \cite{Me96}, it is shown, assuming $g\in C^1$, $V=0$, and an
additional bound on $g$ and its derivative, that for any mass
$M>\|Q\|_{L^2}^2$ and close to $\|Q\|_{L^2}^2$  there exists a blow-up
solution $u$ of  \eqref{NLS} 
such that $\|u_0\|_{L^2}^2=M$. It is also shown that a critical mass
blow-up solution must concentrate at the critical point
$0$. Furthermore, if there exists $\alpha\in ]0,1[$ such that $g$
satisfies  
\begin{equation}
\label{cond_q}
\nabla g(x)\cdot x\le -|x|^{1+\alpha}
\end{equation}
for small $x$, then there is no critical mass solution. Note that this
assumption implies that $g$ is not $C^2$. The existence of minimal
mass blow-up solutions for $g$ which do not satisfy \eqref{cond_q} is
left open in \cite{Me96}. Theorem \ref{thblup} answers positively to
this question for smooth $g$, except in the critical case $\nabla
g(0)=0$ and $\nabla^2g(0)\neq 0$, which includes
the case $\alpha=1$ in \eqref{cond_q}. 
After our article was written, P.~Rapha\"el and J.~Szeftel \cite{RS-p}
have shown the
existence of a minimal mass blow-up solution   in the case where the
matrix $\nabla^2 g(0)$ is non-degenerated.  The strategy of
the proof borrows arguments due to the pioneering works of Y.~Martel
\cite{Ma05}, 
Y.~Martel and F.~Merle \cite{MaMe06}. 
The authors also show a
difficult and strong uniqueness result: this solution is (up to phase
invariance and time translation) the only minimal mass solution. This
is in the spirit of the work by F.~Merle \cite{Me93} for \eqref{eq:nls} (see
also \cite{HmKe05}, and \cite{Ba04} for partial results in the case of
a plane domain) 

Under the assumption $\nabla^2g(0)=0$, the authors of \cite{RS-p}
conjecture that the set of 
minimal mass solutions is parametrized by two additional parameters,
the energy and the asymptotic momentum. Our goal here is to give a
simple construction of critical-mass pseudo-conformal blow-up
solutions in curved geometries (see \S \ref{sec:presentationsurf}) and
we do not 
address the issue of classification of these solutions. 
\smallbreak

We do not address either the question of the existence of 
non-generic blow-up solutions of \eqref{NLS} with supercritical
mass. Examples of such solutions were constructed in
\cite{BoWa97}  for equation \eqref{eq:nls} in space dimensions $1$ and
$2$ and in \cite{KrLeRa09} for the case of Hartree nonlinearity in 
  space dimension $4$. In both cases a supercritical mass blow-up
  solution is obtained, up to a small remainder, as the sum of a minimal mass
  blow-up solution and a solution that vanishes to some order at the
  origin at the blow-up time.   
It should be possible to adapt our method to construct the same type
of solutions. Note that our case is of course simpler than the one of
Hartree-type nonlinearity, where the non-local character of the
nonlinearity appears as an important issue in this construction.  

Let us mention the conjecture, stated in \cite{Pe01},  that there is a 
codimension one submanifold of initial data of equation \eqref{eq:nls}
in $H^1$ leading to pseudo-conformal blow-up. In \cite{KrSc09}
J.~Krieger and W.~Schlag constructed, for this equation in space
dimension $1$, a set of initial data leading to this type of
blow-up. This set is, in spirit, of codimension $1$ in a space
$\Sigma^N$ ($N$ large), without being, rigorously speaking, a submanifold
of this space. The proof of \cite{KrSc09} requires a full use of the
modulations, and also very delicate dispersive estimates for the
linearized operator. 
This type of result is out of reach by our method. As a drawback, the
method of \cite{KrSc09} 
can only deal with functions with a very high regularity, whereas our
fixed point, relying on energy estimates, is essentially at an
$H^{d/2+}$ level. Our argument 
should in particular work in dimensions $d\ge 3$, although the lack
of regularity of the nonlinearity might become an issue in high
dimensions. Let us mention although the works
\cite{KrSc06,KrSc07,Be08,Be09} devoted to the constructions of stable
manifolds around solitons or stationary solutions for other
equations.

\subsection{Strategy of the proof}

The key ingredient of the proof is a result of M.~Weinstein \cite{We85} on
the properties of the linearized NLS operator around the ground state,
which implies that the instability of the linearized equation is only
polynomial, not exponential.

We first consider, as in \cite{BoWa97}, the
pseudo-conformal transformation \eqref{pseudo_conformal}.
Thus $u$ is a solution to \eqref{NLS} on $]0,T[$ if and only if
$\widetilde v$ is 
solution to the following equation on
$\left]\frac{1}{T},+\infty\right[$: 
\begin{equation}
\label{NLSbis}
i\partial_t \widetilde v+\Delta
\widetilde v-\frac{1}{t^2}V\left(\frac{x}{t}\right)\widetilde v+
g\left(\frac{x}{t}\right)|\widetilde v|^{4/d}\widetilde v=0.  
\end{equation}
Intuitively, for large time, the potential term is negligible
(it
belongs to $L^1_tL^\infty_x$, hence it is short range in the sense of
\cite{DeGe97}), and the 
inhomogeneity can be approximated by its value at the
origin. Therefore, a good asymptotic model for \eqref{NLSbis} should
be given by the solution (with the same behavior as $t\to +\infty$) to
the ``standard'' mass-critical focusing nonlinear Schr\"odinger 
equation  \eqref{eq:nls}. 
We want to construct a blow-up solution
to \eqref{NLS} by constructing a solution $\widetilde v$ to
\eqref{NLSbis} which behaves like the solitary wave $e^{it} Q(x)$
(which solves \eqref{eq:nls}) as
$t\to +\infty$. 
In the case $g=1$, there is a huge literature
concerning the existence and stability of solitary waves associated
to \eqref{NLSbis} when the potential $1/t^2V(x/t)$ is replaced by a
\emph{time independent} potential: therefore, these results seem of no
help to study the blow-up phenomenon. 

In a first approximation, we look for a solution of the form
\begin{equation}
  \label{eq:formv}
  \widetilde v=e^{it}(Q+h).
\end{equation}
 Therefore $\widetilde v$ is a solution of \eqref{NLSbis} if and only if 
\begin{equation}
\label{eqh}
i\partial_t h+\Delta
h-h-\frac{1}{t^2}V\left(\frac{x}{t}\right)(Q+h)
+g\left(\frac{x}{t}\right)|Q+h|^{4/d}(Q+h)-Q^{1+4/d}=0.
\end{equation} 
Consider the linearized operator near $Q$
\begin{equation}
\label{defL}
Lf:=-\Delta f+f-\left(\frac{2}{d}+1\right) Q^{4/d}
f-\frac{2}{d}Q^{4/d}\overline{f}. 
\end{equation} 
In \cite{We85}, it is shown that the semi-group $e^{itL}$ is bounded
in the orthogonal space of a $2d+4$ dimensional space $S$, the space
of secular modes, where it grows polynomially. This allows us to
construct the solution $h$ of \eqref{eqh} as a fixed point in a space
of functions that decay polynomially as $t\to +\infty$. Namely, we can
write \eqref{eqh} as   
\begin{equation}
\label{eqh2}
i\partial_t h-Lh = R(h),
\end{equation} 
where $R(h)$ is, roughly speaking, the sum of a source term involving $Q$,
$V$ and $g$, of a similar linear term where $Q$ is replaced by
$h$, and of a term which is  nonlinear in $h$. The latter is
essentially harmless, since we expect $h$ to be small. The first two
terms can be proved small provided that we require a sufficient vanishing
for $V$ and $g-1$ at the origin to balance the polynomial growth of
the semi-group $e^{itL}$ on $S$. This approach is sketched in 
\S\ref{sec:plain} below. 
Note that even though intuitively, it is
natural to expect $1/t^2V(x/t)$ and $g(x/t)-1$ to be negligible for
large time, proving this requires the nontrivial bounds on $e^{itL}$ shown
in \cite{We85}, since
the $S(t)$ behavior is unstable. In the case where $V$ and $g-1$ are
not too flat at the origin, more information is needed.

\smallbreak

In order to loosen the assumptions on the local behavior of $V$ and
$g$ near the origin, we use all the invariances associated to
\eqref{eq:nls} to neutralize as many
secular modes as possible. There is a 
$2d+3$ dimensional family of modulations, given by the scaling,
space-translation, gauge, Galilean, and conformal invariances. By
modulating the function $\widetilde v$ thanks to these 
transformations, we can eliminate all secular modes but one, limiting the
growth of the operator $e^{itL}$. This allows us to decrease the order
to which $V$ 
and $g-1$ vanish at the origin, so as to infer
Theorem~\ref{thblup}. As mentioned before, this 
approach is quite similar in spirit to \cite{KrSc09} for
$L^2$-critical Schr\"odinger equation, and to \cite{Be08,Be09}, where
an $L^2$-supercritical Schr\"odinger equation is considered.

One of the difficulties of our proof is to include the choice of the
modulation parameters in the definition of the operator defining the
fixed point. In this context, the contraction property seems hard to
check: we manage to prove continuity only
(see Proposition~\ref{prop:continuityPhi}). We bypass this difficulty by
using the Schauder fixed point 
theorem. A key step is to obtain energy estimates on an evolution 
equation with a time-dependent operator, which is the sum of the
linearized operator $L$ and a time-dependent perturbative term which
is given by the modulation.

\subsection{Application to NLS on surfaces}
\label{sec:presentationsurf}

Let us first recall that the other known 
blow-up regime, the log-log regime, is not only more stable
on $\R^d$: it is structurally
stable, in the sense that it persists in other geometries. The case of
a domain was settled by F.~Planchon and P.~Rapha\"el \cite{PlRa07},
and the one of a general Riemannian manifold by N.~Burq, P.~G\'erard
and P.~Rapha\"el \cite{BuGeRa08P}. 
\smallbreak

As a consequence of Theorem \ref{thblup}, we are able to construct
blow-up solutions --- with $1/t$ blow-up speed, and with profile related
to $Q$ --- on surfaces flat enough at the blow-up point. To this purpose
we consider rotationally symmetric manifolds. Such a manifold $M$ is a
Riemannian manifold of dimension $2$, given by the metric 
$$ds^2=dr^2+\phi^2(r)\,d\omega^2,$$
where $d\omega^2$ is the metric on the sphere $\mathbf S^{1}$, and
$\phi$ is a smooth function $C^\infty([0,\infty[)$, positive  on
$]0,\infty[$, such that $\phi^{({\rm even})}(0)=0$ and $\phi'(0)=1$. These
conditions on $\phi$ yield a smooth manifold (see e.g.
\cite{Pe98book}). For example, $\R^2$ and
the hyperbolic space $\mathbf H^2$ are such manifolds, with
$\phi(r)=r$ and $\phi(r)=\sinh 
r$, respectively. The volume element is $\phi(r)$, and the distance to
the origin from a point of coordinates $(r,\omega)$ is $r$. Finally,
the Laplace--Beltrami operator on $M$ is  
$$\Delta_M=\partial^2_r+\frac{\phi'(r)}{\phi(r)}\partial_r+
\frac{1}{\phi^2(r)}\Delta_{\mathbf   S^{1}}.$$ 
Now, if we consider $\tilde{u}$ a radial solution of NLS on $M$
(recall that $d=2$) 
\begin{equation}\label{slrev}
i\partial_t \tilde{u}+\Delta_{M}\tilde{u}+|\tilde{u}|^2\tilde{u}=0,
\end{equation}
then the radial function $u$ defined by
$$\tilde{u}(t,r)=u(t,r)\,\left(\frac{r}{\phi(r)}\right)^{1/2}$$
satisfies Equation~\eqref{NLS} with
\begin{equation*}
  V(r)=\frac{1}{2}\frac{\phi''(r)}{\phi(r)}-\frac{1}{4}\,
\left(\left(\frac{\phi'(r)}{\phi(r)}\right)^2-\frac{1}{r^2}\right),
\text{ and }g(r)=\frac{r}{\phi(r)}.
\end{equation*}
Therefore we are in the framework of Theorem~\ref{thblup}, up to
conditions of flatness of the metrics at the blow-up point and of
boundedness of $V$ and $g$ at infinity. These boundedness conditions
corresponds to conditions on the growth of the unit ball volume of the
manifold at infinity.  

This proves the existence of a blow-up solution
of speed $1/t$ and critical mass for such surfaces. 
Notice that the
hyperbolic space $\phi(r)=\sinh r$ correspond to the borderline case
$\partial_r^2 g(0)\neq 0$, which we do not reach with our method. This should be covered, however, by an extension of the work of \cite{RS-p} to equations with a linear potential.
The 
motivation for this case would be to complete the available
information: the virial identity yields a sufficient blow-up condition
which is weaker than in the Euclidean case (\cite{Ba07}), and for
defocusing nonlinearities (or focusing nonlinearities with small
data), the geometry
of the hyperbolic space strongly alters scattering theory, since long
range effects which are inevitable in the Euclidean case, vanish there
(see \cite{BaCaSt08,IoSt09,AnPi09}; see also \cite{BaDu07,BaCaDu09}).

We conclude this subsection by giving explicit examples of surfaces
satisfying the above assumptions.  
\begin{example}[Compact perturbations of the hyperbolic and Euclidian planes]
Let $c_0,d_0\in \R$ and consider $\phi\in C^{\infty}([0,+\infty[)$
such that $\phi(r)=r+c_0 r^5+ \O(r^7)$ as $r\to 0$, and 
$\phi(r)=\sinh(r)+d_0$ or $\phi(r)=r+d_0$ for large $r$.  
Then there exists a solution $\tilde{u}$ of \eqref{slrev} that blows
up at time $t=0$ at the origin $r=0$, and such that $\left\|\nabla
  \tilde{u}(t)\right\|_{L^2}\approx 1/t$ as $t\to 0$. 
An example of such a surface in the case $\phi(r)=r+d_0$ for large $r$
is given by the surface $M$ of $\R^3$ equipped with the induced
Euclidean metric and defined by the equation 
$x=f(y^2+z^2)$, where $f:\R^+\to \R^+$ is a smooth nondecreasing function
such that $f(0)=f'(0)=0$ and $f(s)=x_0>0$ for large $s$.
\end{example}
\begin{remark}
Many simple manifolds do not enter in our framework, as they do not
satisfy the boundedness conditions on $V$ and $g$ at
infinity. Examples are given by the surfaces of $\R^3$ defined by the
equation $x=(y^2+z^2)^{k}$, $k\ge 2$, with the
  induced Euclidean metric, which are spherically symmetric 
manifolds such that $g=r/\phi(r)$ satisfies assumption
\eqref{eq:hypg}, but grows polynomially at infinity. We do not know if
this is only a technical point and it would be interesting, in view of
these examples, to relax the boundedness conditions on $V$ and $g$ at
infinity to a polynomial growth. The case of non-flat compact
surfaces, even with strong symmetry assumptions, is also completely
open. 
\end{remark}

\subsection{Structure of the paper}
\label{sec:structure}

In \S\ref{sec:plain}, we sketch the proof of Theorem~\ref{thblup}
under strong flatness assumptions on $V$ and $g$ near the
origin. The result then follows in a rather straightforward fashion
from a standard fixed point argument, relying on estimates on the
linearized operator $L$ due to M.~Weinstein. In \S\ref{sec:modul}, we
introduce the full family of modulations, in order to reduce the proof
of Theorem~\ref{thblup}. In \S\ref{sec:linear}, we recall some more
precise properties on the linearized operator $L$, which are crucial
for tuning the modulation, as presented in
\S\ref{sec:tuning}. Once the modulation is settled, we study the
non-secular part of the remainder in \S\ref{sec:nonsecular}. The proof
of Theorem~\ref{thblup} is then completed in \S\ref{sec:fixed}, thanks
to compactness arguments. 
Minor technical results are detailed in two appendices, for
the sake of completeness.

\section{Proof of a weaker result}\label{sec:plain}
In this section, we sketch the proof of
Theorem~\ref{thblup} with
\begin{equation*}
  \theta(\tau)=\tau,\quad \l(t)=t,\quad x(t)=0,
\end{equation*}
(hence $\widetilde S=\widetilde S_2$ in Remark~\ref{rem:moreexplicit})
under the 
\begin{hyp}\label{hyp:stronger}
  Let $d=1$ or $2$, and $V,g\in C^\infty(\R^d;\R)$. Assume that for
all $\alpha$, $\d^\alpha g,\d^\alpha V\in L^\infty$, and that there
exist $m_V\ge 7$ and $m_{g}\ge 9$ such that:
\begin{align*}
 \forall |\beta|\le m_V, \quad
 \lvert\partial^{\beta}V(x)\rvert&\le 
C_{\beta}|x|^{m_V-|\beta|}\text{ if }|x|\le 1,\\
\forall |\beta|\le m_{g},\quad
 \lvert\partial^{\beta}(g(x)-1)\rvert&\le 
C_{\beta}|x|^{m_{g}-|\beta|}\text{ if }|x|\le 1.
\end{align*}
\end{hyp}
Recall that the linearized operator $L$ is defined by
\begin{equation*}
Lf:=-\Delta f+f-\left(\frac{2}{d}+1\right) Q^{4/d}
f-\frac{2}{d}Q^{4/d}\overline{f}. 
\end{equation*} 
We will need the following property of $L$, which is a consequence of
\cite{We85} (see also \cite[Proposition
1.38]{Bo99BO}). 
\begin{proposition} 
\label{propWe}
One can decompose $H^1\(\R^d\)$ as $H^1=S\oplus M$, with $S$ (of
finite dimension) and $M$ stable by 
$e^{itL}$ and such that, if $P_M$ and $P_S$ denote the projections on
$M$ and $S$, respectively, the following holds. If $s\ge 1$
and $\psi\in H^s$, then for all $t\ge 1$, 
\begin{align*}
   \left\|e^{itL}P_S(\psi)\right\|_{H^s}& \le C\(1+t^3\)\int
|\psi(x)|e^{-c|x|}dx,\\
\left\|e^{itL}P_M(\psi)\right\|_{H^s}&\le C\|\psi\|_{H^s}.
\end{align*}
Also, if $s'\ge 1$ and $\psi\in \Sigma^{s'}$, then for all $t\ge 1$,
\begin{align*}
  \left\lVert \lvert x\rvert^{s'}e^{itL}P_S(\psi)\right\rVert_{L^2}&
  \le C\(1+t^3\)\int 
  |\psi(x)|e^{-c|x|}dx,\\ 
\left\lVert \lvert x\rvert^{s'}e^{itL}P_M(\psi)\right\rVert_{L^2}&\le
C\||x|^{s'}\psi\|_{L^2}+C\(1+t^{s'}\)\|\psi\|_{H^{s'}}.
\end{align*}
\end{proposition}
In particular, we have for all $s\ge 1$ and $\psi\in H^s\(\R^d\)$,
\begin{equation}
\label{WeS}
\left\lVert e^{itL}\psi\right\rVert_{H^s}\le C\(1+|t|^3\)\|\psi\|_{H^s},
\end{equation}
and for all $\psi\in \Sigma$,
\begin{equation}
\label{Wew}
\left\lVert \lvert x\rvert e^{itL}\psi\right\rVert_{L^2}\le
  C \left\lVert \lvert x\rvert
  \psi\right\rVert_{L^2}+C\(1+|t|^3\)\|\psi\|_{H^{1}}.  
\end{equation}
In order to prove Theorem \ref{thblup}, we need to find a solution of
$$i\partial_t h-Lh =R(h)\quad ;\quad \|h(t)\|_{\Sigma}\Tend t
{+\infty} 0. $$ 
We now give the
expression of $R(h)$: 
$R(h)=R_{NL}(h)+R_L(h)+R_0$, with  
\begin{align*}
R_{NL}(h)&=-g\left(\frac{x}{t}\right)\left[
  |Q+h|^{4/d}(Q+h)-Q^{1+4/d}-\left(\frac{2}{d}+1\right)Q^{4/d}h
  -\frac{2}{d}Q^{4/d}\overline{h}\right],\\ 
R_L(h)&=\left[1-g\left(\frac{x}{t}\right)\right]
\left[\left(\frac{2}{d}+1\right)Q^{4/d}h 
  +\frac{2}{d}Q^{4/d}\overline{h}\right]+\frac{1}{t^2} V
\left(\frac{x}{t}\right)h,\\ 
R_0(t,x)&=\left[1-g\left(\frac{x}{t}\right)\right]Q(x)^{1+4/d}+
\frac{1}{t^2} V\left(\frac{x}{t}\right) Q(x). 
\end{align*}
We construct a fixed point for the functional
\begin{equation}
\label{defM}
\M(h)(t,x)=\int_t^{+\infty}e^{i(\tau-t)L}iR(h)(\tau,x)d\tau,
\end{equation} 
that we decompose as $ \M(h)=\M_{NL}(h)+\M_{L}(h)+\M_0$, in
accordance with the decomposition of $R$. 
Let $s>d/2$ with $s\ge 1$, $T>1$, and $4<b<a$ real numbers to
be chosen later.  
We can prove that $\M$ is a contraction on the ball of radius one
$B_{a,b,T}$ of the space 
$$E_{a,b,T}=\big\{\psi\in C([T,+\infty[;H^s\cap
\Sigma)\,|\;\|\psi\|_{E}<\infty\big\},$$ 
where
$$\|\psi\|_E:= \sup_{t\ge
  T}\(t^a\,\|\psi(t)\|_{H^s}+t^b\,\left\lVert\lvert x\rvert
  \,\psi(t)\right\rVert_{L^2}\).$$  
In the sequel we will denote by $C$ a positive constant, that may
  change from line to line and depend on $a$, $b$, and $s$ but
  \textbf{not on} $T$. 
\smallbreak 

Since the assumptions made in this paragraph are not as general as in
Theorem~\ref{thblup}, we shall only sketch the main steps of the
arguments which lead to the conclusion of Theorem~\ref{thblup}. 

\subsubsection*{Bound on the nonlinear terms}
There exists  $C>0$ such that 
\begin{equation*}
\forall \,h,f\in B_{a,b,T},\quad
\left\|\M_{NL}(h)-\M_{NL}(f)\right\|_{E}\le
\frac{C}{T^{a-4}}\|h-f\|_{E}. 
\end{equation*}
This estimate follows from \eqref{WeS}, \eqref{Wew} and the definition of
$E_{a,b,T}$, which is an algebra embedded in $L^\infty(\R^d)$. Note
also that $4/d\ge 1$, so $R_{NL}$ contains nonlinear terms which are
at least quadratic in $h$. 
\subsubsection*{Bound on the first linear term}
There exists  $C>0$ such that
\begin{equation*}
\forall \,h,f\in B_{a,b,t_0},\quad
\left\|\M_{L}^1(h)-\M_{L}^1(f)\right\|_{E}\leq
\frac{C}{T^{m_{g}-4}}\,\|h-f\|_{E}, 
\end{equation*}
where
\begin{equation*}
  \M_L^1(h)(t,x)=\int_t^{+\infty}e^{i(\tau-t)L}
\,i\left[ g\left(\frac{x}{\tau}\right)-1\right]\left[
  \left(\frac{2}{d}+1\right)Q^{4/d}h
  +\frac{2}{d}Q^{4/d}\overline{h}\right]d\tau.
\end{equation*}
The key remark is that $Q$ decays exponentially. 
If $|x|\le \tau$, by assumption  on $g$,
\begin{equation*}
  \left|\left[g\(\frac{x}{\tau}\)-1\right] 
Q(x)^{4/d}\right|\le C\left|
  g\(\frac{x}{\tau}\)-1\right|e^{-c|x|}\le
\frac{C}{\tau^{m_{g}}}|x|^{m_{g}}e^{-c|x|}\le
\frac{C}{\tau^{m_{g}}}e^{-\frac c2 |x|}. 
\end{equation*}
If $|x|\ge \tau$, in view of the boundedness of $g$ and the
exponential decay of $Q$, 
\begin{equation*}
\left|\left[ g\left(\frac{x}{\tau}\right)-1\right]Q(x)^{4/d}\right|
\le Ce^{-c|x|}\le C e^{-\frac c2 \tau}e^{-\frac c2 |x|}.  
\end{equation*}
Hence the bound
\begin{equation*}
\forall x\in \R^d,\;\forall \tau\ge 1,\quad \left|\left[
    g\(\frac{x}{\tau}\)-1\right]Q(x)^{4/d}\right|\le
\frac{C}{\tau^{m_{g}}}e^{-\frac c2|x|}. 
\end{equation*}
Proceeding along the same lines, we infer
\begin{equation}
\label{boundkHs}
\left\|\left[
    g\left(\frac{\cdot}{\tau}\right)-1\right]Q^{4/d}\right\|_{H^s}\le
\frac{C}{\tau^{m_{g}}},
\end{equation}
and we get by \eqref{WeS} and \eqref{Wew}, the bound on the first linear term.
\subsubsection*{Bound on the second linear term}
There exists  $C>0$ such that for all $h,f\in B_{a,b,T}$, 
\begin{equation*}
 \left\lVert \M_{L}^2(h)-\M_{L}^2(f)\right\rVert_{E}\le
C
\(\frac{1}{T^{m_V-2}}+\frac{1}{T}+\frac{1}{T^{a-b}}\)\|h-f\|_{E},
\end{equation*}
where
\begin{equation*}
  \M_L^2(h)(t,x)=\int_t^{+\infty}e^{i(\tau-t)L}
\(\frac{i}{\tau^2}V\left(\frac{x}{\tau}\right)h(\tau,x)\)d\tau. 
\end{equation*}
 We have, for
$\tau\ge 1$, 
\begin{equation*}
\left\|V\left(\frac{\cdot}{\tau}\right)\right\|_{W^{s,\infty}}\le
C,\quad \text{hence }
\left\lVert V\left(\frac{\cdot}{\tau}\right)h\right\rVert_{H^s}\le
C\,\lVert h\rVert_{H^s}. 
\end{equation*}
Like above, we also have 
\begin{equation}\label{boundVS}
\left\lVert
  V\(\frac{x}{\tau}\)e^{-c|x|}\right\rVert_{H^{s}}\le 
\frac{C}{\tau^{m_V}}. 
\end{equation}
By decomposing $\M_{L}^2(h)$ on its $M$ and $S$ components,  we can
use the estimates of Proposition~\ref{propWe} to get the desired bound on the second linear linear term.

\subsubsection*{Bound on the source term}
There exists  $C>0$ such that 
\begin{equation*}
\left\lVert \M_0\right\rVert_{E}\le
C\(\frac{1}{T^{m_{g}-a-4}}+\frac{1}{T^{m_V-a-2}}\). 
\end{equation*}
This follows easily from  \eqref{boundkHs} and \eqref{boundVS}.

\subsubsection*{Conclusion}
Gathering all the previous estimates together, we have:
\begin{equation}
\label{contract}
\begin{aligned}
&\forall f,h \in  B_{a,b,T},\quad
\left\lVert\M(h)-\M(f)\right\rVert_{E} 
\le C\,{\tt w}(T)\lVert h-f\rVert_{E},\text{ where }\\
&{\tt w}(T) = \frac{1}{T^{a-4}}+
\frac{1}{T^{m_{g}-4}}+\frac{1}{T^{m_V-2}}+ \frac{1}{T}+ 
\frac{1}{T^{a-b}}+ 
\frac{1}{T^{m_{g}-a-4}}+ \frac{1}{T^{m_V-2-a}}.   
\end{aligned}
\end{equation} 
Therefore, for   $m_V>6$ and $m_{g}>8$ (this corresponds to the
assumption made in this paragraph, since $m_V$ and $m_{g}$ are
integers, by regularity of $V$ and $g$), we can choose $a,b$
with $4<b<a$ such that all the powers of $T$ in \eqref{contract} are
positive. Hence we can pick $T$ large enough such that 
\begin{equation}
\label{contract2}
\forall f,h \in B_{a,b,T},\quad \left\lVert 
\M(h)-\M(f)\right\rVert_{E}\le \frac 12\|h-f\|_{E}.
\end{equation}
Taking $f=0$ in \eqref{contract2}, we see that $\M$ maps
$B_{a,b,T}$ into $B_{a,b,T}$. Furthermore, \eqref{contract2} shows
that $\M$ is a contraction on $B_{a,b,T}$, which concludes the
proof of Theorem~\ref{thblup} under Assumption~\ref{hyp:stronger}.

\section{Introducing a modulation}
\label{sec:modul}
We now wish to replace the assumption made in the previous section by
the assumptions of Theorem~\ref{thblup}, which we rewrite: 
\begin{hyp}\label{hyp:gen}
    Let $d=1$ or $2$, and $V\in C^2(\R^d;\R)$, $g\in
    C^4(\R^d;\R)$. Assume that for 
$\d^\beta V\in L^\infty$ for $|\beta|\le 2$, $\d^\beta g\in L^\infty$
for $|\beta|\le 4$ and:
\begin{align*}
 \forall |\beta|\le 1, \quad
 \lvert\partial^{\beta}V(x)\rvert&\le 
C_{\beta}|x|^{1-|\beta|}\text{ if }|x|\le 1,\\
\forall |\beta|\le 3,\quad
 \lvert\partial^{\beta}(g(x)-1)\rvert&\le 
C_{\beta}|x|^{3-|\beta|}\text{ if }|x|\le 1.
\end{align*}
\end{hyp}
At first sight, the above assumption on $V$ is stronger than 
in Theorem~\ref{thblup}. This difference is irrelevant though, in view
of the following remark. For a potential $V$ as in
Theorem~\ref{thblup}, replacing $u(t,x)$ by $u(t,x)e^{itV(0)}$ amounts
to changing $V$ to $V-V(0)$, a potential which satisfies the above
assumption. This
explains the presence of the factor $e^{-itV(0)}$ in the statement of
Theorem~\ref{thblup}. 
\subsection{Modulation and linearization}
As explained in the introduction, we want to obtain a solution to 
\begin{equation}
\label{eq:waveop}
  i\partial_t \widetilde v+\Delta
\widetilde v-\frac{1}{t^2}V\left(\frac{x}{t}\right)\widetilde v+
g\left(\frac{x}{t}\right)|\widetilde v|^{4/d}\widetilde v=0\quad
;\quad \lVert\widetilde v (t)-e^{i\theta(t)}Q\rVert_{\Sigma}\Tend t {+\infty}0,
\end{equation}
where $\theta(t)=t+o(t)$ as $t\to +\infty$. 
Introduce the following modulations:
\begin{equation}
  \label{eq:modul1}
    \widetilde v(t,x) =
  e^{i\(q_1(t) +q_3(t)\cdot x +q_5(t)|x|^2\)}\frac{1}{q_4(t)^{d/2}}v\(\g(t),
\frac{x}{q_4(t)}-q_2(t)\), 
  \end{equation}
with $q_1,q_4,q_5,\g\in \R$ and $q_2, q_3\in \R^d$. The functions $v$
and $\widetilde v$ have similar properties as $t\to +\infty$ if,
morally, 
\begin{equation}
q_1(t),q_2(t),q_3(t),q_5(t)  \Tend t{+\infty} 0\quad ;\quad
  q_4(t)\Tend 
  t {+\infty}1\quad ;\quad \g(t)\Eq t{+\infty} t. 
\end{equation}
We give a rigorous meaning to this line below. 
Note that the second point implies the last one if we assume
\begin{equation*}
  \dot \g(t)=\frac{1}{q_4(t)^2}. 
\end{equation*}
From now on, we define $\g$ as
\begin{equation}\label{eq:gamma}
  \g(t)=\tau_0+\int_{\tau_0}^t\frac{1}{q_4(\sigma)^2}d\sigma,
\end{equation}
where $\tau_0$ is a large time to be determined later. We introduce the new
time and space variables 
\begin{align*}
  (\tau,y)&=\(\g(t),\frac{x}{q_4(t)}-q_2(t)\),\quad \text{or,
    equivalently}\\
  (t,x)&=\Big(\g^{-1}(\tau),q_4\(\g^{-1}(\tau)\)\(y+q_2\(\g^{-1}(\tau)\)
  \)\Big) .
\end{align*}
With our choice for $\g$,  \eqref{eq:waveop} is equivalent to 
\begin{equation}\label{eq:NLSv}
  i\partial_\tau v + \Delta v = V_p v
  - g_p|v|^{4/d}v +iZ_p(v),
\end{equation}
where we have denoted, for $p(\tau)=(p_1,p_2,p_3,p_4,p_5)\in \R\times
\R^d\times\R^d\times \R\times\R$:
\begin{align*}
  iZ_p(v) &= \( p_1+p_3\cdot y +p_5|y|^2\) v +ip_2\cdot \nabla v
  +ip_4\(\frac{d}{2}+y\cdot \nabla\)v,\\
g_p(\tau,y) &= g\(\frac{x}{t}\)=g\(\frac{q_4(y+q_2)}{\g^{-1}(\tau)}\) ,\\
V_p(\tau,y) &=\frac{q_4^2}{t^2}V\(\frac{x}{t}\)=
\frac{q_4^2}{\g^{-1}(\tau)^2}V\(\frac{q_4(y+q_2)}{\g^{-1}(\tau)}\).
\end{align*}
The parameters $q_2$ and $q_4$ are assessed in $\g^{-1}(\tau)$ and
substituting $x= q_4(y+q_2)$,  
\begin{align*}
  p_1& = q_4^2\dot q_1 +q_4^3 \dot q_3\cdot q_2 +q_4^4
  \dot q_5 |q_2|^2 
  +q_4^2 \left|q_3+2q_4 q_5 q_2\right|^2,\\
p_2& =q_4^2 \dot q_2 +q_4 \dot q_4 q_2 -2q_4
a-4q_4^2q_5 q_2,\\ 
p_3& = q_4^3 \dot q_3 +2q_4^4\dot q_5 q_2 +4q_4^3 q_3q_5
+8q_4^4 q_5^2 q_2,\\ 
p_4 &= q_4\dot q_4 -4q_5q_4^2,\\
p_5 &= q_4^4\dot q_5 +4q_4^4 q_5^2. 
\end{align*}
The following rewriting essentially block diagonalizes the above system:
\begin{align*}
p_4 &= q_4\(\dot q_4 -4q_5q_4\).\\
p_5&= q_4^4\(\dot q_5 +4 q_5^2\). \\
p_2&= q_4^2 \dot q_2 -2q_4 q_3 +p_4 q_2.\\
p_3&= q_4^3\(\dot q_3 +4q_3q_5\) +2q_2 p_5.\\
p_1&= p_3\cdot q_2 -p_5 |q_2|^2 +q_4^2\(\dot q_1 +|q_3|^2\). 
\end{align*}
Note  that
we have not examined the asymptotic condition as $t\to+\infty$. We
analyze this aspect more precisely below (see \S\ref{sec:reduced}). 
We write $ v=e^{i\tau}(Q+w)$:
Equation~\eqref{eq:NLSv} is equivalent to 
\begin{equation}
\label{eq_w}
i\partial_\tau w-Lw -iZ_p(w)=iR_p(w)+iZ_p(Q), 
\end{equation}
where $L$ is the linearized operator \eqref{defL}, 
and $R_p(w)=R_{NL}(w)+R_{L}(w)+R_0$
\begin{equation}
  \label{eq:defR}
\left\{
 \begin{aligned}
  iR_{NL}(w) &= -g_p\times \( F\(Q+w\) -F(Q) -\ell(w)\),\\
iR_{L}(w)&= \(1- g_p \)\times  \ell(w)+ V_p w,\\
iR_0&= \(1- g_p \)\times  F(Q)+ V_p Q ,
\end{aligned} 
\right.
\end{equation}
with
\begin{equation*}
 F(z)=|z|^{4/d}z\quad ;\quad \ell(w) =
\(\frac{2}{d}+1\)Q^{4/d}w +\frac{2}{d}Q^{4/d}\overline w.
\end{equation*} 
As we will often write the equation \eqref{eq_w} as $\partial_t w +iL
w=(\ldots)$, 
we also forced a multiplication by $i$ in the definition of $R_p$.
Note that $R_{NL}$, $R_L$ and $R_0$ also depend on the parameter $p$,
although we will usually not indicate it with an index.

The sequel of this section is as follows. In 
\S\ref{sec:ptomod}, we show that one can recover the modulations 
$q_1$,\ldots,$q_5$ and the original variables  $t$  and $x$ from the
parameters $p_1$,\ldots,$p_5$ and the modulated variables $\tau$ and
$y$. In \S\ref{sec:reduced}, we reduce Theorem
\ref{thblup} to the proof of an existence theorem in the modulated
variables $\tau$ and $y$. 

\subsection{From $p$ to the modulation}
\label{sec:ptomod}
From now on, we will work only in the modulated variables $\tau$ and
$y$, and consider, by abuse of notation, the modulations $q_k$ as
functions of $\tau$. 
Denoting by $'$ the derivative with respect to $\tau$, that is $\dot
f=\frac{1}{q_4^2}f'$, the above system reads: 
\begin{equation}
  \label{eq:p/modul'}
  \left\{
\begin{aligned}
p_4 &= \frac{q_4'}{q_4} -4q_5 q_4^{2}.\\
p_5&= q_4^2q_5' +4 q_4^{4}q_5^2. \\
p_2&= q_2' -2q_4 q_3 +p_4 q_2.\\
p_3&= q_4q_3' +4q_4^3q_3q_5 +2q_2 p_5.\\
p_1&= p_3\cdot q_2 -p_5 |q_2|^2 +q_1' +q_4^2|q_3|^2. 
\end{aligned}
\right.
\end{equation}
Recall that we seek $q_4=1+q_{4r}$, with
\begin{equation*}
 q_1,q_2,q_3,q_{4r}, q_5\Tend t{+\infty}0.  
\end{equation*}
Consider these functions as unknowns, to be sought, for $c>0$, in 
\begin{equation}\label{eq:defW}
  W(c,\tau_0)= \{ f\in C([\tau_0,\infty[),\quad
  \|f\|_{c,\tau_0}:=\sup_{\tau\ge \tau_0}\tau^c\lvert 
  f(\tau)\rvert <\infty\}. 
\end{equation}
Our main assumption here is
$  p_j\in W(c(p_j),\tau_0)$, for $1\le j\le 5$.
\begin{lemma}
\label{L:p}
Let $c(p_3)=c(p_5)>2$, $c(p_1)>1$, $c(p_2)=c(p_4)>1$. Then if $\tau_0$
is sufficiently large the following holds. 
Let $p_j\in W(c(p_j),\tau_0)$, $1\le j\le 5$
such that
$$\forall j\in \{1,\ldots,5\},\quad \|p_j\|_{c(p_j),\tau_0}\le 1.$$
 Then there exists a unique family of parameters $q_1,q_2,q_3,q_{4r},
 q_5$, such that the system \eqref{eq:p/modul'}
 holds with
 \begin{itemize}
 \item $q_2,q_{4r}\in W(c(q_2),\tau_0)$ with $\dis
   c(q_2)=\left(\min\left(c(p_5)-2,c(p_4)-1\right)\right)^-$. 
\item $q_3,q_5\in  W(c(q_3),\tau_0)$ with $c(q_3)=c(p_3)/2$.
\item $q_1\in W(c(q_1),\tau_0)$ with
  $c(q_1)=\min\left(c(p_1)-1,c(p_3)-1\right)$, 
 \end{itemize}
and
$$
\|q_{1}\|_{c(q_1),\tau_0}+\|q_{2}\|_{c(q_2),\tau_0}+
\|q_{3}\|_{c(q_3),\tau_0}+\|q_{4r}\|_{c(q_2),\tau_0}+\|q_{5}\|_{c(q_3),\tau_0}\le
1.$$ 
Finally, the variables $(\tau,y)$ and $(t,x)$ are uniformly equivalent:
\begin{equation*}
  \frac{1}{2}\le \frac{d\tau}{dt}\le 2\quad ;\quad
  \frac{1}{2}\<x\>\le \<y\>\le 2 \<x\>. 
\end{equation*}
\end{lemma}
\begin{remark}
\label{R:p}
Under the assumptions of the lemma, we can define implicitly the variable $t$
from the variable $\tau$ in view of the formula \eqref{eq:gamma}.
\end{remark}

\begin{proof}
The first two equations in \eqref{eq:p/modul'} determine $q_{4r}$ and
$q_5$. Then the next two yield $q_2$ and $q_3$, while we infer $q_1$
from the last equation. Thus we first consider
\begin{equation}
  \label{eq:p/mod1}
  \left\{
    \begin{aligned}
      q_{4r}'-4q_5 &=p_4(1+q_{4r})+4q_5
      q_{4r}(3+3q_{4r}+q_{4r}^2),\\
q_5'&= \frac{p_5}{(1+q_{4r})^2}-4(1+q_{4r})^2q_5^2.
    \end{aligned}
\right.
\end{equation}
Introduce the corresponding homogeneous system:
\begin{equation*}
  \frac{d}{d\tau}\(
  \begin{array}[c]{c}
    q_{4r}\\
q_5
  \end{array}
\) = \(
\begin{array}[c]{cc}
  0& 4\\
0&0
\end{array}
\) 
\(
\begin{array}[c]{c}
    q_{4r}\\
q_5
  \end{array}
\).
\end{equation*}
The square of the above matrix is zero, and we infer:
\begin{equation*}
  \exp  \(
\begin{array}[c]{cc}
  0& 4\\
0&0
\end{array}
\) = \(
\begin{array}[c]{cc}
  1& 4\\
0&1
\end{array}
\).
\end{equation*}
Duhamel's formula for \eqref{eq:p/mod1} thus reads:
\begin{align*}
  q_{4r}(\tau)&= -\int_\tau^\infty \Big[ p_4(\sigma)
  \left(1+q_{4r}(\sigma)\right)+4q_5(\sigma)q_{4r}(\sigma)
  \(3+3q_{4r}(\sigma)+q_{4r}^2(\sigma)\)\Big]d\sigma\\   
&\quad
-\int_{\tau}^{+\infty}4(\tau-\sigma)
  \(\frac{p_5(\sigma)}{(1+q_{4r}(\sigma))^2}-4(1+q_{4r}(\sigma))^2
  q_5^2(\sigma) \)d\sigma,\\ 
  q_5(\tau)&=  -\int_\tau^\infty
  \left[\frac{p_5(\sigma)}{(1+q_{4r}(\sigma))^2}-4(1+q_{4r}(\sigma))^2
    q_5^2(\sigma) \right]d\sigma. 
\end{align*}
Denoting $N(k)=\|k\|_{c(k),\tau_0}$, the
first right hand side is controlled by 
\begin{multline*}
   \int_\tau^\infty
    \(\sigma^{-c(p_4)} +\sigma^{-c(q_{4r})-c(q_5)}N(q_{4r})N(q_5)+
    \tau\sigma^{-c(p_5)}+\tau \sigma^{-2c(q_5)}N(q_5)^2\)d\sigma\\ 
    \lesssim
    \tau^{1-c(p_4)}+\tau^{1-c(q_{4r})-c(q_5)}+\tau^{2-\min(c(p_5),2c(q_5))}.   
\end{multline*}
The second right hand side is controlled by
$\tau^{1-\min(c(p_5),2c(q_5))}$. We can solve the above system by a
fixed point argument in the
class that we consider, provided that $\tau_0$ is sufficiently large,
as soon as  
\begin{align*}
&c(q_{4r})+1<c(p_4) &;&\qquad 1<c(q_5),\\
&  c(q_{4r}) +2<\min(c(p_5),2c(q_5)) &;&\qquad c(q_5)+1<\min
  (c(p_5),2c(q_5)). 
\end{align*}
This boils down to
\begin{equation*}
  c(p_4)>1\quad;\quad c(p_5)>2,
\end{equation*}
in which case we may take
\begin{equation*}
  c(q_{4r})= \(\min\(c(p_5)-2,c(p_4)-1\)\)^-\quad ;\quad c(q_5)=
  \frac{1}{2}c(p_5). 
\end{equation*}
Note also that $\tau_0$ can be chosen independent of $p$ such that
$N(p)\le 1$. 
\smallbreak

The system yielding $(q_2,q_3)$ is similar (the constant $4$ becomes a
$2$):
\begin{equation*}
  \left\{
\begin{aligned}
 q_2'-2q_3&=p_2+2q_{4r}q_3-p_4q_2\\
q_3'&=-4(1+q_{4r})^2q_3q_5+\frac{p_3-2q_2 p_5}{1+q_{4r}}.
\end{aligned}
\right.
\end{equation*}
Under the extra assumption $c(p_2)=c(p_4)$ and $c(p_3)=c(p_5)$, we may take
\begin{equation*}
  c(q_2)= c(q_{4r})\quad ;\quad c(q_3)=c(q_5). 
\end{equation*}
It is clear that we may choose $c(q_1)=\min\(c(p_1)-1,c(p_3)-1\)$. 
\smallbreak

The inequalities:
\begin{align*}
 \left\lvert \frac{d}{dt}\(\tau-t\)\right\rvert &= \left\lvert
   \frac{1}{q_4^2}-1\right\rvert= \left\lvert
\(1+q_{4r}\)^{-2}-1 \right\rvert \lesssim \frac{1}{t^{c(q_2)}},\\
\left\lvert y_j-x_j\right\rvert&=\left\lvert \(\frac{1}{q_4}-1\)x_j
  -q_2\right\rvert \lesssim \frac{1}{t^{c(q_2)}}\(|x_j|+1\)
\end{align*}
imply the last part of the lemma.
\end{proof}
The following lemma is a direct consequence of the proof of the
previous result:
\begin{lemma}\label{lem:qLip}
 Let $p$ and $\tilde p$ satisfy the assumptions of
 Lemma~\ref{L:p}. Assume in addition that for all $k$,
 $c(p_k)=c(\tilde p_k)=c(p)>2$. Denote by $q$ and $\tilde q$
 the corresponding  
 modulations provided by Lemma~\ref{L:p}. We have
 \begin{equation*}
   \lvert q_4(\tau) -\tilde q_4(\tau)\rvert + \lvert q_2(\tau)
   -\tilde q_2(\tau)\rvert \lesssim 
   \frac{1}{\tau^{c(p)-2}} \max_{1\le k\le 5} \|p_k-\tilde
   p_k\|_{c(p),\tau_0}. 
 \end{equation*}
\end{lemma}
\begin{proof}
  Subtract the Duhamel's formulations to systems \eqref{eq:p/mod1}
  associated to $p$ and  $\tilde p$, respectively. Denoting
  $e_q(\tau)=|q_4-\tilde q_4|+\tau|q_5-\tilde q_5|$, we have immediately
  \begin{align*}
    e_q(\tau)&\lesssim \tau\int_\tau^\infty
    \(\(\frac{1}{\sigma^{c(p)}}+
    \frac{1}{\sigma^{1+c(q_5)}} 
    \)e_q(\sigma) +\frac{1}{\sigma^{c(p)}} \max_{1\le k\le 5} \|p_k-\tilde
   p_k\|_{c(p),\tau_0}\)d\sigma \\
&\lesssim \tau\int_\tau^\infty
   \( \frac{1}{\sigma^{1+c(q_5)}}
    e_q(\sigma) +\frac{1}{\sigma^{c(p)}} \max_{1\le k\le 5} \|p_k-\tilde
   p_k\|_{c(p),\tau_0}\)d\sigma 
  \end{align*}
From Lemma~\ref{L:p}, $c(q_5)=c(p)/2>1$. We can then apply Gronwall lemma to
$\widetilde e_q(\tau)=e_q(\tau)/\tau$, and the first estimate
follows. The estimate for $q_2$ proceeds along the same lines.  
\end{proof}

\subsection{Reduced problem}
\label{sec:reduced}

In the rest of this paper, we show the following:
\begin{theorem}\label{theo:reduced}
  Let Assumption~\ref{hyp:gen} be satisfied. There exists $\tau_0>0$,
  a modulation $p$ such that $p_j\in W(c(p),\tau_0)$ with
$c(p)>2$, 
 and a solution $w\in C([\tau_0,\infty[;\Sigma)$ to 
\begin{equation}\label{eq:w}
i\partial_\tau w-Lw -iZ_p(w)=iR_p(w)+iZ_p(Q), 
\end{equation}
such that
\begin{equation*}
\|w(\tau)\|_{H^1}\le \frac{C}{\tau^{2^-}},\quad
\lVert \<y\>  w(\tau)\rVert_{L^2}\le \frac{C}{\tau^{1^-}}.
\end{equation*}
\end{theorem}
\begin{proof}[Theorem~\ref{theo:reduced} implies Theorem~\ref{thblup}]
  Writing $v(\tau,y)= e^{i\tau}\(Q(y)+w(\tau,y)\)$, we first see that
  Theorem~\ref{theo:reduced} implies the existence of
  $p=p(\tau)$ like above, 
 and a solution $v\in C([\tau_0,\infty[;\Sigma)$ to 
\begin{gather*}
  i\partial_\tau v + \Delta v = V_p v
  - g_p|v|^{4/d}v +iZ_p(v),\\
\| v  (\tau)-e^{i\tau}Q\|_{H^1}\le \frac{C}{\tau^{2^-}},\quad \lVert \<y\>(  v
  (\tau)-e^{i\tau}Q)\rVert_{L^2}\le \frac{C}{\tau^{1^-}}.
\end{gather*}
If this holds, then Lemma~\ref{L:p} yields a modulation $q$ such that 
\begin{equation}\label{eq:modul0}
  |q_2(t)|+|q_4(t)-1|\Tend t {+\infty} 0, \quad
  |q_1(t)|+|q_3(t)|+|q_5(t)|\le \frac{C}{t^{1^+}}, 
\end{equation}
and a solution of \eqref{eq:waveop},
\begin{equation*}
    \widetilde v(t,x) =
  e^{i\(q_1(t) +q_3(t)\cdot x +q_5(t)|x|^2\)}\frac{1}{q_4(t)^{d/2}}v\(\g(t),
\frac{x}{q_4(t)}-q_2(t)\). 
  \end{equation*}
We now set
\begin{equation*}
  \theta(t)=\gamma(t)\quad ;\quad \l(t) = t q_4\(\frac{1}{t}\)\quad
  ;\quad x(t)= t q_4\(\frac{1}{t}\) q_2\(\frac{1}{t}\). 
\end{equation*}
Equation~\eqref{eq:gamma} and Lemma~\ref{L:p} show that indeed,
$\gamma(t)=t+o(t)$ as $t\to \infty$. We also know from
\eqref{eq:modul0} that
\begin{equation*}
  \l(t)\sim t\quad \text{and} \quad |x(t)|=o(t) \text{ as
  }t\to 0^+ .  
\end{equation*}
In view of the behavior of the $H^1$ and $\F H^1$ norms via the
pseudo-conformal transformation, we readily verify that
Theorem~\ref{thblup} follows from \eqref{eq:modul0}. 
\end{proof}
As suggested by the statement of Theorem~\ref{theo:reduced}, we
construct \emph{simultaneously} the modulation $p$ and the remainder
$w$. We will see in Section~\ref{sec:tuning} that these two unknowns are
related through a nonlinear process.

\section{The linearized operator}
\label{sec:linear}

To prove Theorem~\ref{theo:reduced}, we need more precise properties
concerning the linearized operator $L$ than
those recalled in Proposition~\ref{propWe}. We use again refined
estimates proved in \cite{We85} (see also \cite{BoWa97}). 
\smallbreak

As in \cite{We85}, we identify $\C$ with $\R^2$, and
the space of complex-valued functions $H^1(\R^d,\C)$ with the space
$H^1(\R^d,\R)\times H^1(\R^d,\R)$, 
considering the operator $\LL=iL$ as an operator on $L^2\times L^2$
with domain $H^2\times H^2$: 
\begin{equation*}
\LL=iL=\begin{pmatrix}0 & L_-\\ -L_+ & 0\end{pmatrix}, \quad
L_+=-\Delta+1-\(\frac{4}{d}+1\)Q^{4/d},\quad
L_-=-\Delta+1-Q^{4/d}.
\end{equation*}
Note that $\LL$ is not self-adjoint. We denote by
$$ \< f,g\>=\int_{\R^d} f_1 g_1+\int_{\R^d} f_2 g_2,$$
the scalar product on $L^{2}(\R^d)\times L^{2}(\R^d)$. The space
of secular modes is defined by 
$$S:=\bigcup_{\kappa\ge 1} N\(\LL^\kappa\),$$ 
where $N(A)$ is the null-space of the operator $A$. We next specify
the space $S$ and the dynamics of $e^{itL}$ on $S$. Note that by
direct calculation, 
\begin{equation}
\label{eq_radial_modes}
\left\{
\begin{aligned}
&L_-(|x|^2Q)=-4\(\frac{d}{2}Q+x\cdot \nabla Q\),\quad 
L_-Q=0,\\
-&L_+\(\frac{d}{2}Q+x\cdot\nabla Q\)=2 Q,
\end{aligned}
\right.
\end{equation}
\begin{equation} 
\label{eq_nonradial_modes}
L_-(x_\ell Q)=-2\d_{x_\ell}Q,\quad  L_+(\d_{x_\ell}Q)=0.
\end{equation}
Furthermore, there exists only one radial function $\tQ$ such that
$$L_+ \tQ=-|x|^2Q.$$
Consider for $1\le  \ell\le d$
\begin{gather*}
n_1=-i\alpha_0^{-1} Q; \quad
n_{2,\ell}=-\beta_0^{-1}\partial_{x_\ell}Q,\quad n_{3,\ell}=i\beta_0^{-1} x_\ell
Q\\ 
n_4=\alpha_0^{-1}\(\frac{d}{2}Q+x\cdot\nabla Q\),\quad
n_5=-i\alpha_0^{-1}\(\frac{1}{2}|x|^2Q+\gamma_0 Q\),\quad
n_6=\alpha_0^{-1}\tQ, 
\end{gather*}
(where $\alpha_0,\beta_0,\gamma_0$ are normalization constants,
$\alpha_0,\beta_0>0$). Then 
\begin{equation}
\label{LLL_n}
\left\{
\begin{aligned}
\LL n_1&=\LL n_{2,\ell}=0,& \LL n_4&=-2 n_1,&\LL n_{3,\ell}&=2 n_{2,\ell},\\
\LL n_5&=2n_4,&\LL n_6&=-2n_5+2\gamma_0n_1.&&
\end{aligned}\right.\end{equation} 
This shows that all $n_j$'s are in the space $S$. By similar
computations, the following functions are in the space
$S^*=\bigcup_{\kappa\ge 1} N\(({\LL^*})^\kappa\)$: 
\begin{alignat*}{6}
m_1&=i\tQ, &\ m_{2,\ell}&=x_\ell Q,&\ m_{3,\ell}&=-i\partial_{x_\ell}Q,\\
m_4&=-\frac{1}{2}|x|^2Q-\gamma_0 Q,&\ m_5&=i\frac{d}{2}Q+ix\cdot\nabla
Q,&\ m_6&=-Q.
\end{alignat*}
Moreover,  $M=(S^*)^{\bot}$, and $\< n_k,m_j\>=\delta_{jk}$, so that
$$ P_S h=\sum_{1\le j\le 6} \nu_j n_j,\text{ where }
\nu_j=\< h,m_j\>.$$ 
As a consequence, in view of \eqref{LLL_n}, the exact dynamics of
$e^{itL}$ on $S$ is obtained.

\begin{proposition}
\label{P:We2}
Let $G\in C(\R;H^1\times H^1)$, and $W$ such that
\begin{equation}
\partial_t W+iL W=G
\end{equation}
Denote $\nu_{j}=\< W,m_j\>$ and $d_j=\< G,m_j\>$. Then,
\begin{align*}
\nu'_1&=2\nu_4-2\gamma_0\nu_6+d_1 & \nu'_{2,\ell}&=-2\nu_{3,\ell}+d_{2,\ell} &
\nu'_{3,\ell}&=d_{3,\ell}\\
\nu'_4&=-2\nu_5+d_4&
\nu'_5 &=2\nu_6+d_5&\nu'_6&=d_6.
\end{align*}
\end{proposition}

\section{Tuning the modulation}
\label{sec:tuning}

Our approach consists of a careful examination of \eqref{eq:w}. As we
have seen in the previous section, we can 
write $H^1=M\oplus S$. Recall that
$S$, the generalized kernel of $iL$, is a finite dimensional space,
and that the group $e^{itL}$ is bounded on $M$. To construct the wave
operator of Theorem~\ref{theo:reduced}, we have to control the secular
part of $w$ (its $S$ component). We decompose $w$ into 
$w=w_S+w_M$. By noticing that 
\begin{align}
\label{ZPQ}
Z_p(Q)&=-i\(p_1+p_3\cdot y+p_5|y|^2\)Q+p_2\cdot \nabla
Q+p_4\(\frac{d}{2} +y\cdot \nabla \)Q\\ 
\notag
&=p_1\alpha_0 n_1-p_3\beta_0 n_3+2p_5\alpha_0( n_5-\gamma_0
n_1)-p_2\beta_0 \cdot n_2+\alpha_0 n_4 
\end{align}
is in $S$, we deduce the projected
equations on $S$ and on $M$. Namely, we want to construct a solution
to the system  
\begin{align}
\label{eq_wS}
&\partial_\tau w_S+iLw_S=P_SR_p(w)+P_SZ_p(w)+Z_p(Q),\\
\label{eq_wM}
&\partial_\tau w_M+iLw_M -P_MZ_p(w_M)=P_MR_p(w)+P_MZ_p(w_S).
\end{align}
We introduce
\begin{equation}\label{operator}
  \begin{aligned}
      \Phi(w)(\tau)&= \int_\tau^\infty
  e^{i(\tau-\sigma)L}\(P_SR_p(w)+P_SZ_p(w)+Z_p(Q)\)d\sigma+\Phi_2(w)(\tau)\\
&= \Phi_1(w)(\tau)+ \Phi_2(w)(\tau), 
  \end{aligned}
\end{equation}
where $ \Phi_2(w)=\phi$ is the solution (in $M$ for all $\tau$) of the equation
\begin{equation*}
  \partial_\tau \phi+iL\phi -P_MZ_p(\phi)=P_MR_p(w)+P_MZ_p(w_S).
\end{equation*}
The existence of $ \Phi_2(w)$ will be shown in
\S\ref{sec:nonsecular}. In the
present section, we define the modulation parameter $p$, and estimate
$\Phi_1(w)$. 
The main point in our approach is that $p$  depends on
$w$, and is chosen so that the secular part $\Phi_1(w)$ of $\Phi(w)$ 
belongs to $\vect (n_6)$. As $p$ also appears in
  the definition of $\Phi$ in \eqref{operator}, the 
dependence of $\Phi$ upon $w$ is more implicit (and more nonlinear)
than it may seem. 
\smallbreak

As it is standard, we shall construct  in Section~\ref{sec:fixed} a fixed point
for $\Phi$. However, we shall not use Banach--Picard result (based on
contractions), but rather the Schauder fixed point argument
(based on compactness).

For
$p=(p_1,\ldots,p_5)$, $c>0$, denote, once and for all, 
$$ |p(\tau)|=\max_{1\le k \le 5}|p_k(\tau)|,\quad
\|p\|_{c,\tau_0}=\max_{1\le k \le 5} \|p_k\|_{c,\tau_0}.$$  
The main result of this section is the following:
\begin{proposition}\label{P:p}
Let Assumption~\ref{hyp:gen} be satisfied.
Let $\eps\in ]0,1/3[$. Then if $\tau_0>0$ is large enough we have the
following property.  Let $w\in C([\tau_0,\infty[;H^1)$ with 
\begin{equation}\label{decay_w}
  \sup_{\tau \ge \tau_0}\tau^{2-\eps} \|w(\tau)\|_{H^1}\le 1. 
\end{equation}
There exists a unique
modulation parameter 
$p=p(w)$, such that, for $\tau\ge \tau_0$, 
\begin{equation}\label{pjest}
|p(\tau)|\le \frac{1}{\tau^{3-3\eps}},
\end{equation}
and
\begin{equation*}
\Phi_1(w)(\tau)=\int_\tau^\infty
  e^{i(\tau-\sigma)L}\(P_SR(w)+P_SZ_p(w)+Z_p(Q)\)d\sigma\in \vect{n_6}.
\end{equation*}
Furthermore, for this choice of $p$
\begin{equation}
\label{estimate_m6}
\forall \tau\ge \tau_0,\quad
\left|\<\Phi_1(w)(\tau),m_6\>\right|\le \frac{C}{\tau^{3-2\eps}},
\end{equation} 
where $C$ does not depend on $w$. 
\end{proposition}

We prove Proposition~\ref{P:p} in \S\ref{sec:proj}. We first need some {\it a priori} estimates for arbitrary $p$.

\smallbreak

\subsection{General estimates}
\label{sec:estgen}
Recall from \eqref{eq:defR} the notations:
\begin{equation*}
\left\{
 \begin{aligned}
  iR_{NL}(w) &= -g_p\times \( F\(Q+w\) -F(Q) -\ell(w)\),\\
iR_{L}(w)&= \(1- g_p \)\times  \ell(w)+ V_p w,\\
iR_0&= \(1- g_p \)\times  F(Q)+ V_p Q ,
\end{aligned} 
\right.
\end{equation*}
with
\begin{equation*}
 F(z)=|z|^{4/d}z\quad ;\quad \ell(w) =
\(\frac{2}{d}+1\)Q^{4/d}w +\frac{2}{d}Q^{4/d}\overline w.
\end{equation*} 
\begin{lemma}\label{lem:Rlip}
  Let Assumption~\ref{hyp:gen} be satisfied, and
$$ \|p_k\|_{c(p),\tau_0}\le 1,\; \|\tilde{p}_k\|_{c(p),\tau_0}\le 1,\quad
k\in\{1,\dots ,5\},$$ 
where $c(p)\in ]2,3[$. Then, for every fixed $w$, we have the pointwise
estimates
\begin{gather}
\label{RNL_pointwise} 
  |R_{NL}(w)|\lesssim Q\lvert w\rvert^2+\sum_{3\le j\le 1+4/d}|w|^j,\\
\label{RL_pointwise} 
  |R_L(w)|\lesssim 
    \frac{\<y\>^{3}}{\tau^{3}}e^{-c\<y\>} |w| +
    \frac{1}{\tau^2}\min\Big(1,\frac{\<y\>}{\tau}\Big) |w|,\\
\label{R_pointwise_diff}
 \lvert R_p(w)-R_{\tilde p}(w)\rvert\lesssim \frac{1}{\tau^{c(p)+1}}
  \|p-{\tilde p}\|_{c(p)}\(e^{-c\<y\>} +
  \<y\>^3|w|\<w\>^{4/d}\)\\
  \label{RNL_pointwise_diff}
\lvert R_{NL,p}(w)-R_{NL,\tilde p}(w)\rvert\lesssim \frac{1}{\tau^{c(p)+1}}
  \|p-{\tilde p}\|_{c(p),\tau_0}
  \<y\>^3|w|^2\<w\>^{4/d-1},\\
\label{RL_pointwise_diff}
\lvert R_{L,p}(w)-R_{L,\tilde p}(w)\rvert\lesssim \frac{1}{\tau^{c(p)+1}}
  \|p-{\tilde p}\|_{c(p),\tau_0}
  \<y\>|w|.
\end{gather}
\end{lemma}
\begin{proof}
Estimates \eqref{RNL_pointwise} and
\eqref{RL_pointwise} follow from the definition of $R_{NL}$ and of $R_{L}$ (see
\eqref{eq:defR}), and, for \eqref{RL_pointwise}, from Assumption~\eqref{hyp:gen}. 

Next we estimate $|g_p-g_{\tilde p}|$.  Notice that
\begin{equation*}
\left|\frac{1}{\gamma^{-1}(\tau)}-
\frac{1}{\tilde \gamma^{-1}(\tau)}\right|=
\left|\frac{\tilde \gamma^{-1}(\tau)-
\gamma^{-1}(\tau)}{\gamma^{-1}(\tau)\tilde \gamma^{-1}(\tau)}\right|.   
\end{equation*}
We have
\begin{equation*}
  \frac{d \gamma^{-1}(\tau)}{d\tau} = q_4\(\tau\)^2 =
  \(1+ q_{4r}\(\tau\)\)^2.
\end{equation*}
Therefore, by Lemma~\ref{lem:qLip}:
\begin{equation*}
 \Big\lvert \frac{d}{d\tau}\big( \tilde \gamma^{-1}(\tau)-
\gamma^{-1} (\tau)\big) \Big\rvert \lesssim \left\lvert
q_{4r}\(\tau\) -\tilde q_{4r}
\(\tau\) \right\rvert 
\lesssim \frac{1}{\tau^{c(p)-2}} \|p-\tilde
  p\|_{c(p),\tau_0},
\end{equation*}
 Integrating between $\tau_0$ and $\tau$ and using that
$\gamma(\tau_0)=\tilde \gamma(\tau_0)=\tau_0$
we get, since
$c(p)<3$,  
\begin{equation*}
\left\lvert \tilde
  \gamma^{-1}(\tau)- 
\gamma^{-1}(\tau)\right\rvert\lesssim \frac{1}{\tau^{c(p)-3}}
\|p-\tilde 
  p\|_{c(p),\tau_0}. 
\end{equation*}
This rather poor estimate yields the more interesting one
\begin{equation*}
  \left|\frac{1}{\gamma^{-1}(\tau)}-
\frac{1}{\tilde \gamma^{-1}(\tau)}\right|\lesssim \frac{1}{\tau^{c(p)-1}}
\|p-\tilde 
  p\|_{c(p),\tau_0}. 
\end{equation*}
Denote $\lambda = q_4/\gamma^{-1}$, and $\tilde \lambda$ its
counterpart associated to $\tilde p$. We can write
\begin{equation*}
  g_p(\tau,y)-g_{\tilde p}(\tau,y)= g(\lambda y +\lambda q_2)-
  g(\tilde \lambda y +\tilde\lambda\tilde q_2 ).  
\end{equation*}
Note that Assumption~\ref{hyp:gen} implies
\begin{equation*}
  \lvert g(a)-g(b)\rvert \lesssim \lvert a-b\rvert\(|a|^2+|b|^2\). 
\end{equation*}
Invoking Lemma~\ref{L:p} and Lemma~\ref{lem:qLip}, we deduce
\begin{align*}
  \lvert g_p(\tau,y)-g_{\tilde p}(\tau,y) \rvert &\lesssim
  \(\lvert \lambda -\tilde \lambda\rvert \lvert y\rvert +\lvert
  \lambda q_2-\tilde \lambda \tilde q_2\rvert\)\(\lambda^2+\tilde
  \lambda^2\)\<y\>^2 \\
&\lesssim \frac{1}{\tau^{c(p)-1}}
 \|p-{\tilde 
  p}\|_{c(p),\tau_0}\<y\> \times \frac{1}{\tau^2}\<y\>^2\\
 &\lesssim \frac{1}{\tau^{c(p)+1}}
\|p-{\tilde 
  p}\|_{c(p),\tau_0}\<y\>^3. 
\end{align*}
We have a similar estimate on
$V_{p}-V_{\tilde{p}}=\lambda^2V(\lambda(y+q_2))-
\tilde{\lambda}^2V(\tilde{\lambda}(y+\tilde{q}_2))$: 
\begin{equation*}
  \lvert V_p(\tau,y)-V_{\tilde p}(\tau,y) \rvert \lesssim \frac{1}{\tau^{c(p)+1}}
\|p-{\tilde 
  p}\|_{c(p),\tau_0}\<y\>.
\end{equation*}
By definition, we have (without splitting the terms as in
  \eqref{eq:defR})
  \begin{equation*}
    -iR_p(w) = g_p \times \lvert Q+w\rvert^{4/d}(Q+w) - V_p\times
    (Q+w)-F(Q)-\ell (w).
  \end{equation*}
We also have
$$ \left|Q+w\right|^{1+4/d}\lesssim Q^{1+4/d}+|w|^{1+4/d},$$
and the estimate \eqref{R_pointwise_diff} follows. 

Estimates \eqref{RNL_pointwise_diff} and \eqref{RL_pointwise_diff} of
the lemma are a straightforward consequence 
of the definitions \eqref{eq:defR} of $R_{L,p}$ and $R_{NL,p}$, and of
the above estimates.  
\end{proof}

We introduce the notation, for $1\le j\le 6$, 
\begin{equation}\label{Dj}
D_j(p)(\tau)=\<P_SR_p(w)+P_SZ_p(w),m_j\>.
\end{equation}
\begin{lemma}\label{lem:5.1}
 Let Assumption~\ref{hyp:gen} be satisfied. If 
$$ \|p\|_{c(p),\tau_0}\le 1,$$
where $c(p)>2$, then we have for all
$\tau\ge \tau_0$, 
\begin{equation}
\label{Djest}
\begin{aligned}
 |D_j(p)(\tau)|
\lesssim &\  \|w\|_{H^1}^{1+4/d} +
\|w\|_{H^1}^2+\frac{1}{\tau^2}\|w\|_{L^2}+|p(\tau)| 
\|w\|_{L^2}\\
&+
\begin{cases}
0&\text{ if } j=2,4,6\\
1/{\tau^{1+c(p)^-}}+
1/{\tau^4}&\text{ if }j=1,5\\ 
1/{\tau^3}&\text{ if }j=3.
\end{cases} 
\end{aligned}
\end{equation}
\end{lemma}
\begin{proof}

Taking the $L^2$-norm in $y$ in the pointwise estimate
\eqref{RNL_pointwise}, Sobolev embedding yields: 
\begin{equation*}
  \lVert  R_{NL}(w)(\tau)\rVert_{L^2}\lesssim \sum_{2\le k\le
    1+4/d} \lVert w(\tau)\rVert_{H^1}^k. 
\end{equation*}
By the pointwise estimate \eqref{RL_pointwise} we get
\begin{equation*}
  \lVert  R_{L}(w)(\tau)\rVert_{L^2}\lesssim \frac{1}{\tau^2}\lVert
  w(\tau)\rVert_{L^2}. 
\end{equation*}These estimates yield, since $m_j\in \Sch\(\R^d\)$,
\begin{equation*}
  \lvert D_j(p)(\tau)\rvert \lesssim \sum_{2\le k\le
    1+4/d} \lVert w(\tau)\rVert_{H^1}^k + \frac{1}{\tau^2}\lVert
  w(\tau)\rVert_{L^2}+\lvert
  \<R_0,m_j\>\rvert + \lvert p(\tau)\rvert \lVert
  w(\tau)\rVert_{L^2}, 
\end{equation*}

Notice that $R_0$ is purely imaginary, that $m_2$, $m_4$ and $m_6$ are
real, and thus 
\begin{equation*}
\forall j\in \{2,4,6\}\quad \<R_0,m_j\>=0,
\end{equation*}
which yields the first case in \eqref{Djest}.

By Assumption \ref{hyp:gen}, the Taylor expansion of $g$ near the origin
reads:
\begin{equation*}
  g(x)= 1 + \sum_{|\alpha|=3} c_{\alpha} x^{\alpha}+\O(|x|^4). 
\end{equation*}
In view of Lemma~\ref{L:p}, we infer
\begin{align*}
  g_p(\tau,y) &= 1+\frac{q_4^3}{\(\gamma^{-1}(\tau)\)^3}
  \sum_{|\alpha|=3} c_{\alpha} \(y + q_2\)^\alpha + \O\(\frac{\lvert
    q_4\(y+q_2\)\rvert^4}{\tau^4}\)\\
&= 1+\frac{q_4^3}{\(\gamma^{-1}(\tau)\)^3}
  \sum_{|\alpha|=3} c_{\alpha} y^\alpha +\O\(\frac{\<y\>^2}{\tau^{3+c(q_2)}}\)+
  \O\(\frac{\<y\>^4}{\tau^4}\). 
\end{align*}
Notice that if $j\in\{1,5\}$, $m_j$ is a radial function. Thus if $|\alpha|=3$, 
$$ \int y^{\alpha} m_j=0.$$
Arguing similarly on $V$, we infer 
\begin{equation*}
  \lvert \<P_S R_0,m_j\>\rvert = \lvert
  \<R_0,m_j\>\rvert=\lvert\<(1-g_p)F(Q)+V_pQ,m_j\>\rvert \lesssim 
  \frac{1}{\tau^{3+c(q_2)}} +\frac{1}{\tau^4}.
\end{equation*}
Lemma~\ref{L:p} then yields the second
case in \eqref{Djest}. To prove the third case, we use the pointwise estimate
\begin{equation*}
  \lvert R_0\rvert\lesssim \<\frac{y}{\tau}\>^{3} Q^{1+4/d} +
  \frac{1}{\tau^{2}}\({\bf 
  1}_{|y|>\tau} +\frac{\<y\>}{\tau}{\bf
  1}_{|y|\le \tau}\) Q .
\end{equation*}
Since $Q$ decays exponentially, this yields
\begin{equation*}
  \lVert  R_0(\tau)\rVert_{L^2}\lesssim \frac{1}{\tau^{3}},
\end{equation*}
and the third case in \eqref{Djest} follows. 
\end{proof}

\subsection{Control of the secular modes by projection}
\label{sec:proj}

We next prove Proposition~\ref{P:p}.
We introduce, for arbitrary $p$, 
\begin{equation}\label{dj}
d_j(p)(\tau)=\<P_SR_p(w)+P_SZ_p(w)+Z_p(Q),m_j\>=D_j(p)(\tau)+\<Z_p(Q),m_j\>.
\end{equation}
By the explicit expression \eqref{ZPQ} of $Z_p(Q)$ we get the
relations between $d_j$ and $D_j$:  
\begin{align*}
d_1(p)& =D_1(p)+\alpha_0p_1-2\alpha_0\gamma_0p_5, &
d_2(p)&=D_2(p)-\beta_0p_2,& d_3(p)&=D_3(p)-\beta_0 p_3\\ 
d_4(p)&=D_4(p)+\alpha_0p_4,&d_5(p)&= D_5(p)+2\alpha_0p_5,&d_6(p)&=D_6(p),
\end{align*}
where $\alpha_0,\beta_0,\gamma_0$ are real constants, $\alpha_0,\beta_0>0$.
From \eqref{decay_w}, we know that $w$ tends to zero as
$\tau\to+\infty$. Recalling that $Z_p(Q)\in S$ for any parameter $p$, 
the stability of $S$ by $e^{itL}$ shows that $ \Phi_1(w) \in S$. 
Denote, as in Proposition~\ref{P:We2}, 
\begin{equation*}
  \Phi_1(w)(\tau)=\sum_{j=1}^6 \nu_j(\tau) n_j.
\end{equation*}
By Proposition~\ref{P:We2},
$$\nu_6(\tau)=-\int_{\tau}^{+\infty} d_6=-\int_{\tau}^{+\infty} D_6(p),$$
which is well-defined in view of \eqref{Djest}, \eqref{decay_w}, \eqref{pjest}.
We want $\nu_j$ to vanish, for $1\le j\le 5$, so in view of
Proposition~\ref{P:We2}, we would like to impose
\begin{equation*}
  d_2=d_3=d_4 =0\quad ;\quad d_5=-2\nu_6 \quad ;\quad d_1 = 2\g_0\nu_6.
\end{equation*}
The proposition  follows if we get a fixed point $p$ in the unit ball
of $(W(3-3\eps,\tau_0))^{3+2d}$ (the space $W$ is defined by
\eqref{eq:defW}) for the operator 
$\Psi(p)=\widetilde{p}=(\widetilde{p}_1,\widetilde{p}_2,\widetilde{p}_3,
\widetilde{p}_4,\widetilde{p}_5)$: 
\begin{align*}
  \widetilde{p}_5 &= \frac{1}{2\alpha_0}\(-D_5(p) + 2\int_\tau^\infty
  D_6(p)\)\quad ;\quad \widetilde{p}_4 
    =-\frac{1}{\alpha_0} D_4(p) \\  \widetilde{p}_j &=
    \frac{1}{\beta_0}D_j(p),\ 
    j=2,3 \quad ;\quad 
\widetilde{p}_1=- \frac{1}{\alpha_0}D_1(p) 
-\frac{\g_0}{\alpha_0} D_5(p) . 
\end{align*}
Let $B$ be the closed unit ball in $(W(3-3\eps,\tau_0))^{3+2d}$. We first
show that $B$ is stable by $\Psi$. By \eqref{decay_w}, and since
$0<\eps<1/2$, we have, for 
$\tau\ge \tau_0\gg 1$, 
\begin{align}
\label{premiere_D}
\|w\|_{H^1}^{1+4/d} +&
\|w\|_{H^1}^2
+\frac{1}{\tau^2}\|w\|_{L^2}+|p(\tau)|
\|w\|_{L^2}\\
\notag
&\lesssim 
\frac{1}{\tau^{(1+4/d)(2-\eps)}}+\frac{1}{\tau^{4-2\eps}}
+\frac{1}{\tau^{4-\eps}}+\frac{1}{\tau^{5-4\eps}}\le 
\frac{1}{\tau^{4-2\eps}}.
\end{align}
By definition of $\Psi$, \eqref{premiere_D}
and the estimates \eqref{Djest} on $D_j$, we get, for $j\in\{2,3,4\}$ 
\begin{align*}
|\widetilde{p}_j(\tau)|&\lesssim |D_j(p)(\tau)|\lesssim \|w\|_{H^1}^{1+4/d} +
\|w\|_{H^1}^2
+\frac{1}{\tau^2}\|w\|_{L^2}+\frac{1}{\tau^3}+|p(\tau)|
\|w\|_{L^2}\\
&\le \frac{C}{\tau^3}\le \frac{1}{\tau^{3-3\eps}},
\end{align*}
if $\tau\ge \tau_0$ and $\tau_0$ is chosen sufficiently large. In view
of the estimates  
\eqref{Djest} and \eqref{premiere_D}, we have
\begin{equation}
\label{int_D6}
 \int_{\tau}^{+\infty}|D_6(p)(\sigma)|d\sigma\lesssim
\int_{\tau}^{+\infty}\frac{1}{\sigma^{4-2\eps}}d\sigma 
\lesssim \frac{1}{\tau^{3-2\eps}},
\end{equation} 
provided $\tau_0\gg 1 $. 
By the estimate \eqref{Djest} (second case) and \eqref{premiere_D} we
get, taking again $\tau\ge \tau_0\gg 1$, 
\begin{equation*}
  |\widetilde{p}_5(\tau)|\lesssim
|D_5(p)(\tau)|+
\int_{\tau}^{+\infty}|D_6(p)(\sigma)|d\sigma\lesssim
\frac{1}{\tau^{4-2\eps}}+\frac{1}{\tau^{4-3\eps}}+\frac{1}{\tau^{3-2\eps}}\le 
\frac{1}{\tau^{3-3\eps}},
\end{equation*}
and similarly
\begin{equation*}
  |\widetilde{p}_1(\tau)|\lesssim |D_1(p)(\tau)|+ 
|D_5(p)(\tau)|\le \frac{1}{\tau^{3-3\eps}}.
\end{equation*}
As a consequence $\widetilde{p}=\Psi(p)\in B$, and the
stability property of $\Psi$ is settled. \

It remains to prove the contraction property of $\Psi$, 
\begin{equation}\label{contractionp}
\lVert \Psi(p)-\Psi({\tilde p})\rVert_{3-3\eps,\tau_0}\le
\kappa\lVert p-{\tilde p}\rVert_{3-3\eps,\tau_0},
 \end{equation}
for all $p,{\tilde p}\in B$, with $\kappa<1$. 
In view of the definition of
$\Psi$, it is enough to show that if $\iota$ is small, and $\tau_0$ is
chosen large enough, we have, for $\tau\ge \tau_0$, 
\begin{equation}\label{contrp1}
\|D_j(p)-D_j({\tilde p})\|_{3-3\eps,\tau_0}\le 
\iota\|p-{\tilde p}\|_{3-3\eps,\tau_0},
\end{equation}
for $1\le j\le 5$,  and
\begin{equation}\label{contrp2}
\|D_6(p)-D_6({\tilde p})\|_{4-3\eps,\tau_0}\le \iota
\|p-{\tilde p}\|_{3-3\eps,\tau_0}.
\end{equation}
Recall that by definition, 
$$D_j(p)=\<P_SR_p(w)+P_SZ_p(w),m_j\>.$$
We have
\begin{align*}
  \left\lvert \<P_SZ_p(w)-P_SZ_{\tilde p}(w),m_j\>
\right\rvert&\lesssim |p_k(\tau)-{\tilde p}_k(\tau)|\|w(\tau)\|_{L^2}\\
&\lesssim  \frac{1}{\tau^{2-\eps}}|p_k(\tau)-{\tilde
  p}_k(\tau)|\lesssim \frac{1}{\tau^{5-4\eps}}\|p-{\tilde 
  p}\|_{3-3\eps,\tau_0}. 
\end{align*}
By the pointwise estimate \eqref{R_pointwise_diff} we get
\begin{equation*}
  \left\lvert \< R_p(w)-R_{\tilde p}(w),m_j\>\right\rvert 
 \lesssim \frac{1}{\tau^{4-3\eps}}\|p-{\tilde 
  p}\|_{3-3\eps,\tau_0}.
\end{equation*}
Taking  $\tau_0$ larger if necessary, we deduce the estimate
\eqref{contrp1}. 

To prove \eqref{contrp2}, we argue as in the proof of
Lemma~\ref{lem:5.1}: 
\begin{equation*}
  D_6(p)=\<R_{NL}(w)+R_L(w) +P_SZ_p(w),m_6\>,
\end{equation*}
that is, the contribution of $R_0$ vanishes, since $R_0$ is purely
imaginary and $m_6$ is real. We then invoke inequalities
\eqref{RNL_pointwise_diff} and \eqref{RL_pointwise_diff} of
Lemma~\ref{lem:Rlip}, to infer: 
\begin{align*}
  \left\lvert \< R_p(w)-R_{\tilde p}(w),m_6\>\right\rvert 
 &\lesssim \frac{1}{\tau^{4-3\eps}}\|p-{\tilde 
  p}\|_{3-3\eps,\tau_0}\|w(\tau)\|_{H^1}\<\|w(\tau)\|_{H^1}\>^{4/d}\\
&\lesssim
\frac{1}{\tau^{4-3\eps}}\times \frac{1}{\tau^{2-\eps}}\|p-{\tilde 
  p}\|_{3-3\eps,\tau_0},
\end{align*}
which gives \eqref{contrp2} and concludes the proof of the contraction property \eqref{contractionp}.

Therefore there exists a fixed point $p\in B$ for $\Psi$. For this $p$,
we have $\nu_j(\tau)=0$, for $1\le j\le 5$. Moreover, since 
$$\Phi_1(w)(\tau)=\nu_6(\tau) n_6,$$
it remains to show \eqref{estimate_m6}, that is, to check that
$|\nu_6(\tau)|\lesssim 
1/\tau^{3-2\eps}$. This follows immediately from \eqref{int_D6}
and the fact that ${\nu}_6'=D_6$. 
\qed

\section{The non-secular part}
\label{sec:nonsecular}

As announced in the previous paragraph, we now study the
$M$-component of $w$, which has to solve \eqref{eq_wM}. For this, we
consider the operator $\Phi_2$, that is, we study the equation
\begin{equation}
  \label{eq:Mgen}
  \d_\tau \phi +iL\phi -P_M Z_p(\phi)=F\quad ;\quad \|
    \phi(\tau)\|_{\Sigma}\Tend \tau {+\infty} 0, 
\end{equation}
where $F\in C([\tau_0,\infty[;M)$. 
For $a,b>0$, let 
\begin{align*}
  X(a,b,\delta)& =\left\{ \phi \in C([\tau_0,\infty[;M\cap
    \Sigma^\delta),\quad
    \|\phi\|_{X(a,b,\delta)}<\infty\right\},\text{ where}\\
\|\phi\|_{X(a,b,\delta)}&= \sup_{\tau\ge \tau_0}\tau^a\left\lVert  
  \phi(\tau)\right\rVert_{H^\delta}
+\sup_{\tau\ge \tau_0}\tau^b \left\lVert 
  \<y\>^{\delta}\phi(\tau)\right\rVert_{L^2}.
\end{align*}

The main result of this section is:

\begin{proposition}\label{prop:cauchyM}
Let $\tau_0>0$ and $p\in C([\tau_0,\infty[)^{3+2d}$ such that
\begin{equation*}
  \forall \tau\ge \tau_0,\quad |p(\tau)|\le
  \frac{1}{\tau^{3-3\eps}}.
\end{equation*}
Assume that $F\in X(a+1+\eta,b+1+\eta,\delta)$, with $a,b>0$,
$\eta>0$ and 
\begin{equation*}
 \delta<a-b<\delta(2-3\eps),\quad  1\le \delta\le 5.
\end{equation*}
Then
\eqref{eq:Mgen} has a unique solution $\phi\in 
X(a,b,\delta)$. Furthermore, it satisfies 
\begin{equation*}
  \|\phi\|_{X(a,b,\delta)}\le \mu \|F\|_{X(a+1+\eta,b+1+\eta,\delta)}.
\end{equation*}
\end{proposition}
\subsection{Energy estimates}
\label{sec:energyest}
Recall the important property, established in \cite{We85}:
on $M$, the $H^1$ norm $\lVert\cdot\rVert_{H^1}$ is equivalent to
$\lVert\cdot\rVert_M$, where 
\begin{equation*}
  \| \phi\|_M^2 = \re \<L\phi,\phi\>.
\end{equation*}
\begin{lemma}\label{lem:apriori}
  Let $\kappa\in\N$, and $F\in L^1([\tau_0,\infty[;
  \Sigma^{2\kappa+1})$. Suppose that $\phi\in 
  C([\tau_0,\infty[;M\cap \Sigma^{2\kappa+1})$ solves
  \eqref{eq:Mgen} and tends to $0$ in $\Sigma^{2\kappa+1}$ as $\tau\to
  +\infty$. There exists $C>0$ such that for 
  all $\tau\ge \tau_0$, the following holds:
\begin{align*}
\|\phi(\sigma)\|_{H^{2\kappa+1}} \le
C\int_\tau^\infty&\Big( \|F(\sigma)\|_{H^{2\kappa+1}}\\
+&|p(\sigma)|
\(\|\phi(\sigma)\|_{H^{2\kappa+1}} + \|\<y\>\<\nabla\>^{2\kappa}
\phi(\sigma)\|_{L^2}\)\Big)d\sigma,\\
\|\<y\>^{2\kappa+1}\phi(\sigma)\|_{L^2}\le
C\int_\tau^\infty&\Big(\|\<y\>^{2\kappa+1}F(\sigma)\|_{L^2}+
\|\<y\>^{2\kappa}\nabla 
\phi(\sigma)\|_{L^2}+\|\phi(\sigma)\|_{L^2}\\ 
& +
|p(\sigma)|\|\< y\>^{2\kappa+1}
\phi(\sigma)\|_{L^2}\Big)d\sigma. 
\end{align*}
\end{lemma}
\begin{proof}
We begin with the first inequality in the case $\kappa=0$. 
Multiply \eqref{eq:Mgen} by $L\overline \phi$, integrate with
respect to $y$ and consider the real part: 
\begin{equation*}
  \re \int_{\R^d} \d_\tau
\phi\,L\overline{\phi}-\re\int_{\R^d}P_MZ_p(\phi)\,L\overline{\phi} 
=\re \int_{\R^d} F L\overline{\phi}.
\end{equation*}
We readily check the identity
\begin{equation*}
   \re \int_{\R^d} \d_\tau
\phi\,L\overline{\phi} =\frac{1}{2}\frac{d}{d\tau}\|\phi\|_M^2.
\end{equation*}
A straightforward integration by parts yields
\begin{equation*}
  \left\lvert \int_{\R^d} F L\overline{\phi}\right\rvert \lesssim
  \lVert F\rVert_{H^1}\lVert \phi\rVert_{H^1}.
\end{equation*}
It remains to estimate 
\begin{equation}\label{decomp}
\re\int_{\R^d}P_MZ_p(\phi)\,L\overline{\phi}=
\re\int_{\R^d}Z_p(\phi)\,L\overline{\phi}-
\re\int_{\R^d}P_SZ_p(\phi)\,L\overline{\phi}. 
\end{equation}
We start with the first term. Recall that
$$ Z_p(\phi)=-i\(p_1+p_3\cdot y+p_5|y|^2\)\phi+p_2\cdot \nabla \phi
+p_4\left(\frac d2 +y\cdot \nabla\right)\phi$$
and
$$ L \overline{\phi}= -\Delta \overline{\phi}+\overline{\phi}-
\(\frac{2}{d}+1\)Q^{4/d} \overline{\phi}-\frac{2}{d}Q^{4/d}\phi.$$
We have, by elementary integration by parts:
\begin{gather*}
\re i\int p_1\phi\,\Delta \overline{\phi}=0,\\
\left|\re i\int p_3 \cdot y \phi \Delta \overline{\phi}\right|=
\left|\im \int \phi \,p_3 \cdot  \nabla \overline{\phi}\right|\le
|p_3|\|\phi\|_{L^2}\|\phi\|_{H^1},\\ 
\left|\re i\int  p_5 |y|^2\phi \Delta \overline{\phi}\right|=
2\left|p_5\im \int \phi\,y\cdot \nabla \overline{\phi}\right|
\le 2|p_5|\|\phi\|_{H^1}\|\<y\>\phi\|_{L^2},\\
\re \int p_2\cdot\nabla \phi\Delta \overline{\phi}=0,\\
\re \int p_4 \(\frac d2 +y\cdot \nabla\)\phi\Delta
\overline{\phi}=
-p_4 \int |\nabla \phi|^2.
\end{gather*}
We infer:
\begin{equation*}
\left|\re \int Z_p(\phi)\Delta\overline{\phi} \right|
\le C(|p_3|+|p_4|)\|\phi\|_{H^1}^2+|p_5|\|\phi\|_{H^1}\|\<
y\>\phi\|_{L^2}. 
\end{equation*} 
We easily deduce that the first term in \eqref{decomp}
is controlled by 
\begin{equation*}
\Big|\re \int Z_p(\phi)L\overline{\phi} \Big|
\le C|p|\|\phi\|_{H^1}(\|\phi\|_{H^1}+\|\<y\>\phi\|_{L^2}).
\end{equation*}
For the remaining second term in \eqref{decomp} we use the structure
of the space $S$, 
$$\left|\re\int P_SZ_p(\phi)\,L\overline{\phi}\right|=\left|\sum_{1\le j\le
    6}\< Z_p(\phi), m_j\>\re\int n_j\,L\overline{\phi}\right|.$$ 
Integrating by parts both in the scalar product and in the integral, we get
$$\left|\re\int P_SZ_p(\phi)\,L\overline{\phi}\right|\le C|p|\|\phi\|_{L^2}^2.$$
Summarizing, we have obtained
$$\frac{d}{d\tau}\|\phi\|_{M}^2\le C\|F\|_{H^1}\|\phi\|_{H^1}+
C|p|\|\phi\|_{H^1}\big(\|\phi\|_{H^1}+\|\<y\>\phi\|_{L^2}\big).$$
Since the $M$-norm and the $H^1$-norm are equivalent on $M$, the
first inequality of the lemma follows in the case $\kappa=0$.
\smallbreak

Let $\kappa \ge 1$. We write \eqref{eq:Mgen} as 
\begin{equation*}
\d_{\tau} \phi+iL \phi -Z_p(\phi)=F-P_SZ_p(\phi).
\end{equation*}
Applying the operator $(iL)^\kappa$ we get
\begin{equation*}
\d_{\tau} \((iL)^\kappa\phi\)+iL \((iL)^\kappa \phi\)
-(iL)^\kappa Z_p(\phi)=(iL)^\kappa F-(iL)^\kappa P_SZ_p(\phi). 
\end{equation*}
Hence
\begin{equation*}
\d_{\tau} \((iL)^\kappa\phi\)+iL \((iL)^\kappa \phi\) 
-Z_p(iL)^\kappa\phi=(iL)^\kappa F+\left[ (iL)^\kappa,Z_p\right]\phi 
-(iL)^\kappa P_SZ_p(\phi),
\end{equation*}
where $\left[(iL)^\kappa,Z_p\right]$ denotes the commutator of the operators
$(iL)^\kappa$ and $Z_p$. 
By direct computation, the commutator $[iL,Z_p]$ is an operator of
order $2$ in $(\<y\>,\nabla)$, which is only of order $1$ in
$\left<y\right>$ and whose coefficients are multiples of
$p_1$,\ldots,$p_5$:
\begin{align*}
  [iL,Z_p]\phi &= \left[iL, -i\(p_1 +p_3\cdot y +p_5|y|^2\)+p_2\cdot \nabla
    +p_4\(\frac{d}{2}+y\cdot \nabla\)\right]\phi\\
&= -i\left[\Delta, -i\(p_3\cdot y +p_5|y|^2\)+p_4 y\cdot
  \nabla\right]\phi\\
&\quad -i\(\frac{2}{d}+1\)\left[Q^{4/d}, p_2\cdot \nabla
    +p_4 y\cdot \nabla\right]\phi
-i\frac{2}{d}\left[Q^{4/d}, p_2\cdot \nabla
    +p_4 y\cdot \nabla\right]\overline \phi\\
&= -2\(p_3\cdot \nabla +2p_5 y\cdot \nabla +dp_5+ip_4\Delta\)\phi\\
&\quad +i\(\frac{2}{d}+1\)\(p_2\cdot \nabla\(Q^{4/d}\)+p_4 y\cdot
\nabla\(Q^{4/d}\)\)\phi\\
&\quad +i\frac{2}{d}\(p_2\cdot \nabla\(Q^{4/d}\)+p_4 y\cdot
\nabla\(Q^{4/d}\)\)\overline\phi.
\end{align*}
Furthermore
$$ \left[(iL)^\kappa,Z_p\right]=\sum_{j=0}^{\kappa-1}
(iL)^j\left[iL,Z_p\right](iL)^{\kappa-j-1}.$$ 
Hence
\begin{equation*}
  \left\lVert \left[(iL)^\kappa,Z_p\right]\phi \right\rVert_{H^{1}}\le 
|p|\,\(\|\phi \|_{H^{2\kappa+1}}+\left\lVert \<y\>
  \<\nabla\>^{2\kappa}\phi\right\rVert_{L^{2}}\).
\end{equation*}
Notice also that
\begin{equation*}
  \left\|(iL)^\kappa P_S(Z_p(\phi)) \right\|_{H^1} 
=\Big\|\sum_{1\le j\le 6}\<Z_p(\phi),m_j\>(iL)^\kappa
  n_j\Big\|_{H^1} \lesssim  |p|\,\|\phi\|_{L^2}.
\end{equation*}
Denoting $\phi_\kappa= (iL)^\kappa \phi$, we see that
$\phi_{\kappa}(\tau)\in M$ and that it solves 
\begin{equation*}
  \d_\tau \phi_\kappa +iL\phi_\kappa -P_M Z_p(\phi_\kappa)=(iL)^\kappa
  F+ \left[(iL)^\kappa,Z_p\right]\phi 
-(iL)^\kappa P_SZ_p(\phi)+P_SZ_p(\phi_\kappa). 
\end{equation*}
From the case $\kappa=0$ and the previous estimates, we get:
\begin{align*}
\left\|(iL)^\kappa \phi(\tau)\right\|_{H^1}&\lesssim
\int_{\tau}^{+\infty}\Big(\|F(\sigma)\|_{H^{2\kappa+1}}\\
&\quad +\lvert p(\sigma)\rvert
  \Big(\lVert \phi(\sigma)\rVert_{H^{2\kappa+1}}+
\big\|\<y\>\<\nabla\>^{2\kappa}\phi(\sigma)\big\|_{L^2}\Big)\Big)d\sigma. 
\end{align*}
Noting that for a large constant $K$, depending on $\kappa$, we have for all $f\in M$,
$$ \left\|(iL)^\kappa f\right\|_{H^1}+K\|f\|_{H^1}\approx  \|f\|_{H^{2\kappa+1}},$$
and using the case $\kappa=0$ 
to bound $\|\phi(\tau)\|_{H^1}$, we get the first estimate of the lemma.  
\smallbreak

To conclude, we estimate the momenta: for $s\in \N$, we compute more
generally 
\begin{align*}
\frac 12\frac{d}{d\tau}\int \< y\>^{2s} |\phi|^2&=\re \int \< y\>^{2s}\d_\tau\phi
\overline{\phi} =\re \int \< y \>^{2s}(-iL\phi+P_MZ_p(\phi)+F)
\overline{\phi}\\  
&=\im \int \< y\>^{2s} L\phi\,\overline{\phi}+\re \int \<
y\>^{2s}P_MZ_p(\phi)  \,\overline{\phi}+\re \int \<y\>^{2s} F\overline{\phi}. 
\end{align*}
By a direct integration by parts
\begin{equation}
\label{momentumB}
\left|\im \int \< y\>^{2s} L\phi\, \overline{\phi}\right|\lesssim \|
\< y\>^{s-1}\nabla \phi\|_{L^2}\|\< y\>^s \phi\|_{L^2}+\|\phi\|_{L^2}^2. 
\end{equation}
Furthermore,
$$\re \int \< y\>^{2s} P_MZ_p(\phi)  \,\overline{\phi}=\re \int \< y\>^{2s}
Z_p(\phi)  \,\overline{\phi}-\re \int \< y\>^{2s} P_SZ_p(\phi)
\,\overline{\phi}.$$  
On the one hand,
\begin{align*}
\left|\re \int \< y\> ^{2s}  Z_p (\phi)\,  \overline{\phi}\right|&=\left|\re
  \int \< y\>^{2s}  p_2\cdot \nabla \phi\, \overline{\phi}+p_4\< y\>^{2s}
  \left(\frac d2 + y\cdot \nabla\right)\phi\, \overline{\phi} \right| \\
&\lesssim \max_{k=2,4}\lvert p_k\rvert\lVert \<  y\>^s\phi\rVert_{L^2}^2.
\end{align*}
On the other hand,
\begin{equation*}
 \left|\re \int \<  y\> ^{2s} P_SZ_p (\phi)\,
  \overline{\phi}\right|=\left|\sum_{1\le j\le 6}\< Z_p(\phi),
  m_j\>\re\int n_j\,\<  y\> ^{2s}\overline{\phi}\right| 
\lesssim \lvert p\rvert\lVert \phi\rVert_{L^2}^2.
\end{equation*}
Hence
\begin{equation}
\label{momentumC}
\left|
\re \int \< y\>^{2s} P_MZ_p(\phi)  \,\overline{\phi}\right|\lesssim
\lvert p\rvert\lVert \<  y\>^s\phi\rVert_{L^2}^2. 
\end{equation}
Combining \eqref{momentumB}, \eqref{momentumC}, we obtain the second
estimate of the lemma, concluding this proof. 
\end{proof}

\subsection{Refined a priori estimates}
\label{sec:aprioriref}
In the sequel, we consider $0<\eps<1/3$, and extra smallness
assumptions will be precised when needed.
\begin{lemma}\label{lem:aprioriref}
  Let $\tau_0>0$ and $p\in C([\tau_0,\infty[)^{3+2d}$ such that
\begin{equation*}
  \forall \tau\ge \tau_0,\quad |p(\tau)|\le
  \frac{1}{\tau^{3-3\eps}}.
\end{equation*}
Assume that $F\in X(a+1+\eta,b+1+\eta,\delta)$, with 
$\eta>0$ and 
\begin{equation}
\label{a_b_delta}
\delta<a-b<\delta(2-3\eps),\quad  \delta\in \{1,5\},
\end{equation}
where $X$ is defined in Proposition~\ref{prop:cauchyM}. 
Let $\mu>0$. If $\tau_0$ is 
sufficiently large, 
every solution $\phi\in X(a,b,\delta)$ of \eqref{eq:Mgen} satisfies
\begin{equation}
\label{estimation_phi}
  \|\phi\|_{X(a,b,\delta)}\le \mu \|F\|_{X(a+1+\eta,b+1+\eta,\delta)}.
\end{equation}
\end{lemma}
\begin{remark}
  The restriction $\delta\in\{1,5\}$ in the above statement is arbitrary. 
\end{remark}
\begin{proof}
\noindent {\bf First case: $\delta=1$.} Denote by
\begin{equation*}
  M_1=\|F\|_{X(a+1+\eta,b+1+\eta,1)}.
\end{equation*}
The $H^1$-estimate and the momentum estimate with $\kappa=0$ of
Lemma~\ref{lem:apriori} read, along with 
the assumption on $p$:
\begin{align*}
  \|\phi(\tau)\|_{H^1} \le
C\int_\tau^\infty&\Big(\frac{M_1}{\sigma^{a+1+\eta}}+
\frac{1}{\si^{3-3\eps}}\(\|\phi(\sigma)\|_{H^1}+\|\<
y\> \phi(\sigma)\|_{L^2}\)\Big)d\sigma,\\ 
\|\<y\>\phi(\tau)\|_{L^2} \le
C\int_\tau^\infty& \Big(\frac{ M_1}{\sigma^{b+1+\eta}}+
\frac{1}{\si^{3-3\eps}}\|\< y\> \phi(\sigma)\|_{L^2}
 +\|\phi(\sigma)\|_{H^1}\Big)d\sigma.
\end{align*}
We apply Lemma~\ref{L:diff} with the following data:
\begin{equation*}
  \alpha_1=\beta_1=0\ ;\ \alpha_2=\beta_2=1\ ;\
  a_1=a_2=b_1=2-3\eps\ ;\ b_2=-1. 
\end{equation*}
This is possible under the assumptions $a,b>0$ and  $1<a-b<2-3\eps$,
which are fulfilled in the context of Lemma~\ref{lem:aprioriref}. We
then have 
\begin{equation*}
  \|\phi\|_{X(a,b,1)}\le \mu
  \|F\|_{X(a+1+\eta,b+1+\eta,1)},
\end{equation*}
for $\tau_0$ sufficiently large.
\smallbreak

\noindent {\bf Second case: $\delta=5$.}  Denote by
\begin{equation*}
  M_5=\|F\|_{X(a+1+\eta,b+1+\eta,5)}.
\end{equation*}
To proceed in a similar way as in the first case, we use interpolation
estimates \eqref{eq:inter8} and \eqref{eq:inter9}. By Lemma
\ref{lem:apriori} in the case $\kappa=2$, we obtain 
\begin{align*}
  \|\phi(\sigma)\|_{H^5} \le
C\int_\tau^\infty&\Big( \frac{M_5}{\sigma^{a+1+\eta}}\\
&+\frac{1}{\si^{3-3\eps}}
\(\|\phi(\sigma)\|_{H^5} + \|\<y\>^5
\phi(\sigma)\|_{L^2}^{1/5}\|\phi(\sigma)\|_{H^5}^{4/5}\)\Big)d\sigma,\\ 
\|\<y\>^5\phi(\sigma)\|_{L^2}\le
C\int_\tau^\infty&\Big(\frac{
  M_5}{\sigma^{b+1+\eta}}+ 
\|\<y\>^5\phi(\sigma)\|_{L^2}^{4/5}\|\phi(\sigma)\|_{H^5}^{1/5}
+\|\phi(\sigma)\|_{L^2}\\  
& +
\frac{1}{\si^{3-3\eps}}\|\< y\>^5 \phi(\sigma)\|_{L^2}\Big)d\sigma.
\end{align*}
We apply Lemma~\ref{L:diff} with the following data:
\begin{equation*}
  \alpha_1=\beta_3=0\ ;\ \beta_2=1\ ;\ \alpha_2=\beta_1=\frac{1}{5}\
  ;\ 
  a_1=a_2=b_3=2-3\eps\ ;\ b_1=b_2=-1. 
\end{equation*}
This is possible under the assumptions $a,b>0$ and  $5<a-b<10-15\eps$,
which are fulfilled in the context of Lemma~\ref{lem:aprioriref}. We
then have 
\begin{equation*}
  \|\phi\|_{X(a,b,5)}\le \mu
  \|F\|_{X(a+1+\eta,b+1+\eta,5)},
\end{equation*}
for $\tau_0$ sufficiently large.
Summarizing, we have obtained the lemma in the following cases:
\begin{alignat*}{4}
1<a-b<2-3\eps&\text{ and }&\delta=1,\\
5<a-b<5(2-3\eps)&\text{ and }&\delta=5,
\end{alignat*}
which corresponds to the announced result. 
\end{proof}

\subsection{Proof of  Proposition~\ref{prop:cauchyM}} 
\label{sec:cauchyM}

The proof is set up in the same spirit as the existence of M\o ller's
wave operators. Let $\chi(\tau)=1-H(\tau)$, where $H$ is the Heaviside
function, be the function equal to $1$ for $\tau<0$ and $0$ for
$\tau>0$. We first consider the case where $\delta=5$ and 
$F\in L^1([\tau_0,\infty[;M\cap \Sigma^5)$. 
For $(\tau_n)_n$ a sequence going to $+\infty$, consider 
\begin{equation}
  \label{eq:Mgentronque}
  \d_\tau \phi_n +iL\phi_n -P_M Z_p(\phi_n)=\chi(\tau-\tau_n)F \quad ;\quad 
    \phi_{n\mid \tau=1+\tau_n}= 0. 
\end{equation}
To begin with, we remove the projection $P_M$ from the left hand side, and
consider
\begin{equation}
  \label{eq:Mgentronque2}
  \d_\tau \phi_n +iL\phi_n- Z_p(\phi_n)=\chi(\tau-\tau_n)F\quad ;\quad 
    \phi_{n\mid \tau=1+\tau_n}= 0. 
\end{equation}
We show that for every $n$, 
  \eqref{eq:Mgentronque2} has a unique solution 
  $\phi_n\in C([\tau_0,\infty[;\Sigma^5)$.
To see this, remove the modulation by reversing the approach presented
in \S\ref{sec:modul}: recalling \eqref{eq:modul1}, define $\widetilde
\phi_n$ by  
\begin{equation*}
  \widetilde \phi_n(t,x) =
  e^{i\(q_1(t) +q_3(t)\cdot x +q_5(t)|x|^2\)}\frac{1}{q_4(t)^{d/2}}\phi_n\(\g(t),
\frac{x}{q_4(t)}-q_2(t)\),
\end{equation*}
where $\gamma$ is given by \eqref{eq:gamma} and the $q_j$'s are
well-defined function of the $p_k$'s in view of Lemma~\ref{L:p}. We
check that \eqref{eq:Mgentronque2} is then equivalent to an equation of
the form
\begin{equation*}
  i\d_t \widetilde \phi_n +\Delta\widetilde \phi_n =
  W_1\widetilde \phi_n 
  +W_2\overline{\widetilde \phi_n} +\widetilde F_n\quad ;\quad 
    \widetilde \phi_{n\mid t= t_n}=0,
\end{equation*}
where the notation $\widetilde F_n$ is obvious,
$t_n=\gamma^{-1}(\tau_n+1)$, and the potentials are given by 
\begin{align*}
  W_1(t,x)&= \frac{1}{q_4(t)^2}\(1-\(\frac{2}{d}+1\)
Q\(\frac{x}{q_4(t)}-q_2(t)\)^{4/d}\),\\
 W_2(t,x)&= -\frac{2}{dq_4(t)^2
}Q\(\frac{x}{q_4(t)}-q_2(t)\)^{4/d}e^{2i\(q_1(t) +q_3(t)\cdot x
  +q_5(t)|x|^2\)} .
\end{align*}
We note that $W_j\in L^{\infty}([t_0,\infty[;W^{5,\infty}(\R^d))$,
$j=1,2$. We can then construct $\widetilde \phi_n$ in
$C([t_0,\infty[;L^2)$: a fixed point argument yields $\widetilde
\phi_n$ on small time 
intervals (with a non-trivial initial data in order to repeat the
process), and we can split $[t_0,t_n]$ into finitely many time
intervals on which we can control the $L^\infty_tL^2_x$-norm of
$\widetilde \phi_n$ by the $L^1_tL^2_x$-norm of $\widetilde F_n$ on the
same time interval. We can proceed along the same line to construct
$\widetilde \phi_n$ in $C([t_0,\infty[;H^5)$, and then infer that
$\widetilde \phi_n$ is also in $C([t_0,\infty[;\Sigma^5)$ (with
$\widetilde \phi_{n\mid t\ge t_n}=0$). We skip the easy details. 
\smallbreak

We deduce that \eqref{eq:Mgentronque2} has a unique solution 
  $\phi_n\in C([\tau_0,\infty[;\Sigma^5)$. The case of
  \eqref{eq:Mgentronque} follows easily, by rewriting it as
\begin{equation*}
  \d_\tau \phi_n +iL\phi_n - Z_p(\phi_n)=-P_SZ_p(\phi_n)+
  \chi(\tau-\tau_n)F\quad ;\quad  
    \phi_{n\mid \tau=1+\tau_n}= 0, 
\end{equation*}
and by recalling that
\begin{equation*}
  P_SZ_p(\phi_n) = \sum_{j=1}^6 \<Z_p(\phi_n),m_j\>n_j,\text{ hence }
  \|P_SZ_p(\phi_n)(\tau)\|_{\Sigma^5}\lesssim
  \frac{1}{\tau^{3-3\eps}}\|\phi_n(\tau)\|_{L^2}. 
\end{equation*}
The important point which we must note now is that $\phi_n\in
C([t_0,\infty[;M\cap \Sigma^5)$, which is compactly supported in time,
has no secular part. This is 
so thanks to Proposition~\ref{P:We2}, and the integral formulation of
\eqref{eq:Mgentronque}, which can be written as:
\begin{equation*}
  \phi_n(\tau) = \int_\tau^{1+\tau_n}
  e^{i(\sigma-\tau)L}\((\chi(\sigma-\tau_n)F(\sigma)+ P_M 
  Z_p(\phi_n)(\sigma)\) d\sigma,\quad \tau\ge \tau_0.
\end{equation*}
Since $\chi(\cdot-\tau_n)F\in L^1([\tau_0,\infty[;M\cap \Sigma^5)$,
Proposition~\ref{P:We2} 
shows that the right hand side of the above equation has no
non-trivial $S$-component. Therefore, $\phi_n(\tau)\in M$. 
\smallbreak

To conclude, we note that under the assumptions of the proposition,
$\chi(\cdot-\tau_n)F$ converges to $F$ in
$X(a+1+\eta/2,b+1+\eta/2,\delta)$. Since \eqref{eq:Mgen} is 
linear,
Lemma~\ref{lem:aprioriref} shows that $\phi_n$ is a Cauchy sequence in
$X(a,b,\delta)$, thus it converges in this space to $\phi$ solution
to \eqref{eq:Mgen} which satisfies \eqref{estimation_phi}. Uniqueness 
follows from Lemma~\ref{lem:aprioriref}, and we have defined an
operator $F\mapsto \phi$.

By density and Lemma~\ref{lem:aprioriref}, the result remains true if
we assume only $\delta=1$ (and $F\in
X(a+1+\eta,b+1+\eta,\delta)$). The proposition then follows by complex
interpolation between the cases $\delta=1$ and $\delta=5$.\qed

\section{Fixed point argument}
\label{sec:fixed}

In this section we show  Theorem~\ref{theo:reduced}.

Recall that we have defined the operator $\Phi$ as follows:
\begin{equation*}
  \begin{aligned}
      \Phi(w)(\tau)&= \int_\tau^\infty
  e^{i(\tau-\sigma)L}\(P_SR_p(w)+P_SZ_p(w)+Z_p(Q)\)d\sigma+\Phi_2(w)(\tau)\\
&= \Phi_1(w)(\tau)+ \Phi_2(w)(\tau), 
  \end{aligned}
\end{equation*}
where $ \Phi_2(w)=\phi$ is the solution (in $M$) of the equation
\begin{equation*}
  \partial_\tau \phi+iL\phi -P_MZ_p(\phi)=P_MR_p(w)+P_MZ_p(P_Sw)
\end{equation*}
given by Proposition \ref{prop:cauchyM}.
The modulation $p$ is a function of $w$ itself, defined in
\S\ref{sec:tuning}, Proposition \ref{P:p}. To prove
Theorem~\ref{theo:reduced} (hence 
Theorem~\ref{thblup}), we show that $\Phi$ has a fixed point in a
suitable space. Consider 
for $0<\eps<1/3$ and $1<\delta$ 
  \begin{equation*}
    Y(\delta,\eps,\tau_0)=\Big\{ w\in C\left([\tau_0,\infty[;M\cap
    \Sigma^\delta\right)+C\left([\tau_0,\infty[;\vect n_6\right);
  \|w\|_{\delta,\eps,\tau_0}<\infty\Big\} 
\end{equation*}
where $\|w\|_{\delta,\eps,\tau_0}$ is defined as
\begin{equation*}
\sup_{\tau\ge\tau_0} \tau^{2-\eps}\|P_M
w(\tau)\|_{H^\delta}+\sup_{\tau\ge\tau_0} \tau^{2-2\eps-\delta}\|\<y\>^\delta P_M 
    w(\tau)\|_{L^2}+\sup_{\tau\ge\tau_0} 
\tau^{3-3\eps}\lvert \<w(\tau),m_6\>\rvert. 
\end{equation*}

\subsection{Stability}
\label{sec:stability}

The main result of this section is the following: 
\begin{proposition}\label{prop:stabilite}
  Let $\delta \in ]1,2[$, and $0<\eps<1/4$ so that
  $\eps<1-\delta/2$.
There exists $\tau_0>0$ 
such that $\Phi$ maps
the closed
unit ball of $Y(\delta,\eps,\tau_0)$ to itself. 
\end{proposition}
\begin{proof}
For $w\in Y(\delta,\eps,\tau_0)$, Proposition~\ref{P:p} yields a
modulation $p$ such that 
\begin{equation*}
\sup_{\tau\ge \tau_0}\tau^{3-3\eps}
|p(\tau)|\le 1. 
\end{equation*}
By Proposition~\ref{prop:cauchyM}, 
$P_S\Phi(w)=\Phi_1(w)$. Since by
  Proposition~\ref{P:p},
$\Phi_1(w)\in \vect (n_6)$, the secular part of $\Phi(w)$ has the
suitable structure for $Y$. Moreover, by \eqref{estimate_m6}, for
  $\tau\ge \tau_0$,
  \begin{equation*}
    \lvert
    \<\Phi_1(w)(\tau),m_6\>\rvert\le 
    \frac{C}{\tau^{3-2\eps}}\le \frac{C}{\tau_0^\eps}\frac{1}{\tau^{3-3\eps}}. 
  \end{equation*}
 Therefore,
increasing $\tau_0$ if necessary, 
\begin{equation*}
  \sup_{\tau\ge\tau_0} \tau^{3-3\eps}\lvert \<\Phi(w)(\tau),m_6\>\rvert\le
\frac 12.
\end{equation*}
Thus $\Phi_1(w)$ is in the $1/2$-ball of $Y(\delta,\eps,\tau_0)$.
\smallbreak

To control the non-secular part $P_M\Phi(w)=\Phi_2(w)$, we apply
Proposition~\ref{prop:cauchyM} with 
\begin{equation*}
  F = P_M R_p(w)+P_MZ_p\(P_S w\). 
\end{equation*}
We look for $a$ and $b$ such that $F\in X(a+1+\eta,b+1+\eta,\delta)$. 
We note that
\begin{equation*}
  P_MZ_p\(P_S w\) = \<w,m_6\> P_M Z_p(n_6),
\end{equation*}
so we have the estimate
\begin{equation*}
  \|P_MZ_p\(P_S w\)(\tau)\|_{\Sigma^\delta}\lesssim \lvert
  \<w(\tau),m_6\>\rvert\, 
  |p(\tau)|\lesssim \frac{1}{\tau^{3-3\eps}}\times
  \frac{1}{\tau^{3-3\eps}}=\frac{1}{\tau^{6-6\eps}}. 
\end{equation*}
The delicate term, which explains the assumption $\delta<2$, is the
last one, $P_M R_p(w)$. We treat separately the contributions of
$R_0$, $R_L$ and $R_{NL}$.   
Since $d\le 2$ and $\delta>1$, $H^\delta
(\R^d)$ is an algebra, and we infer
\begin{align*}
  \|P_M R_{NL}(w)\|_{H^\delta}&\le \|R_{NL}(w)\|_{H^\delta}+ \|P_S
  R_{NL}(w)\|_{H^\delta} \\
&\lesssim \|w\|_{H^\delta}^2
  +\|w\|_{H^\delta}^{1+4/d}\lesssim \frac{1}{\tau^{2(2-\eps)}}.  
\end{align*}
From the pointwise estimate \eqref{RNL_pointwise},
\begin{align*}
  \|\<y\>^\delta P_M R_{NL}(w)\|_{L^2}&\le \|\<y\>^\delta
  R_{NL}(w)\|_{L^2}+ 
\|\<y\>^\delta P_S R_{NL}(w)\|_{L^2} \\
&\lesssim \|w\|_{L^2}^2
  +\lVert\<y\>^{\delta}w\rVert_{L^2}\sum_{2\le j\le
    4/d}\|w\|_{H^{\delta}}^{j} \lesssim 
  \frac{1}{\tau^{2(2-\eps)}}. 
\end{align*}
We next treat the contribution of $R_L$. Using that $\tau^2 V_p$ is
bounded in the Sobolev space $W^{2,\infty}$, uniformly for $\tau\geq
1$, we get 
\begin{equation*}
  \|P_M R_{L}(w)\|_{H^\delta}\lesssim
  \frac{1}{\tau^2}\|w\|_{H^\delta}\lesssim \frac{1}{\tau^{4-\eps}}.
\end{equation*}
Using simply the boundedness of the external potential $V$, we infer
\begin{equation*}
\|\<y\>^\delta P_M R_{L}(w)\|_{L^2}\lesssim
\frac{1}{\tau^3}\|w\|_{L^2} + \frac{1}{\tau^2}\|\<y\>^\delta w\|_{L^2}
\lesssim \frac{1}{\tau^{4-2\eps-\delta}}.
\end{equation*}
The term $P_MR_0$ can be estimated in a similar way, up to the fact
that the $H^\delta$-norm and the momenta of $Q$ do not decay in time:
\begin{equation*}
  \|P_M R_0\|_{H^\delta}\lesssim \frac{1}{\tau^3}\quad ;\quad
  \|\<y\>^\delta P_M R_0\|_{L^2}\lesssim \frac{1}{\tau^3}.
\end{equation*}
Summarizing, we have obtained 
\begin{gather*}
\|F\|_{H^\delta}\lesssim
\frac{1}{\tau^{6-6\eps}}+\frac{1}{\tau^{4-2\eps}}+
\frac{1}{\tau^{4-\eps}}+ \frac{1}{\tau^3}\lesssim \frac{1}{\tau^3},\\ 
\|\<y\>^\delta F\|_{L^2}\lesssim
\frac{1}{\tau^{6-6\eps}}+\frac{1}{\tau^{4-2\eps}}+
\frac{1}{\tau^{4-2\eps-\delta}}+ \frac{1}{\tau^3}\lesssim
\frac{1}{\tau^{4-2\eps-\delta}}, 
\end{gather*}
meaning that $F\in X(3,4-2\eps-\delta,\delta)$.
We can then apply Proposition~\ref{prop:cauchyM} provided that there
exists $a,b,\eta>0$ with
\begin{equation*}
 \delta<a-b<\delta(2-3\eps),
\end{equation*}
such that  $F\in X(a+1+\eta,b+1+\eta,\delta)$. We take $a+1+\eta =3$
(note that this constraint comes from $R_0$). This requires $a
=2-\eta$ ($\eta>0$ can be arbitrarily small), and since on the other
hand, we must have $a>\delta$, this explains why we have assumed
$\delta<2$. By taking 
$\eta=\frac{\eps}{2}$, we get as a constraint
\begin{equation*}
 \delta<\eps+\delta<\delta(2-3\eps).
\end{equation*}
As $\eps<1/4$ and $\delta>1$, this condition is fulfilled. Therefore, by Proposition~\ref{prop:cauchyM},
$\Phi_2(w)\in X(2-\frac{\eps}{2},2-\frac{3}{2}\eps-\delta,\delta)$. By
increasing $\tau_0$ if necessary, $\Phi_2(w)$ is also in the
$1/2$-ball of  $Y(\delta,\eps,\tau_0)$, and the proposition follows.  
\end{proof}

\subsection{Compactness}
\label{sec:compactness}
We recall the following compactness result, which is a particular case
of \cite[Corollary 4]{Si87}. 
\begin{theorem}[From \cite{Si87}]\label{T:Simon}
Let $X\subset B\subset Y$ be Banach spaces such that $X$ is compactly
embedded into $B$, and $B$ is continuously embedded into $Y$. Let
$\tau_0<\tau_1$ and $F$ be a subset of $L^{\infty}([\tau_0,\tau_1];X)$
such that $\left\{\frac{\partial v}{\partial \tau},\; v\in F\right\}$
is bounded in $L^{\infty}([\tau_0,\tau_1];Y)$. Then $F$ has compact
closure in $C([\tau_0,\tau_1];B)$. 
\end{theorem}
Fix $\eps$, $\delta$ as in Proposition \ref{prop:stabilite}.
Let $K$ be the closed unit ball of $Y(\delta,\eps,\tau_0)$. By
Proposition \ref{prop:stabilite}, the operator $\Phi$ maps $K$ into
itself. 
  Notice that $K$ is closed into $Y(\delta',\eps',\tau_0)$ if 
  $\delta'<\delta$, $\eps'>\eps$ and $2\eps+\delta<2\eps'+\delta'$. 
In this subsection we show the following lemma. 

\begin{lemma}
\label{L:compactness}
Let $0<\eps<\eps'<1/4$ and $1<\delta'<\delta<2$ 
and assume
$2\eps+\delta<2\eps'+\delta'<2$  (this implies that  $Y(\delta,\eps,\tau_0)$ is
continuously embedded into   $Y(\delta',\eps',\tau_0)$).  
The image $\Phi(K)$ of $K$ has compact closure in  $Y(\delta',\eps',\tau_0)$. 
\end{lemma}
\begin{remark}
The assumptions of the lemma are satisfied for example by $\delta,\eps,\delta',\eps'$ defined by
$$ \delta=2-4\eps, \quad \eps'=2\eps,\quad \delta'=2-5\eps,$$
for some small $\eps>0$.
\end{remark}
\begin{proof}
It is sufficient to show that for all $r>0$, there exists a finite
number $N$ of functions in $\psi_n\in Y(\delta',\eps',\tau_0)$, such
that  
\begin{equation}
\label{precompact}
\psi\in \Phi(K)\Longrightarrow \exists n\in \{1,\ldots,N\},\;
\left\|\psi-\psi_n\right\|_{\delta',\eps',\tau_0}<r. 
\end{equation}
Recall first that $\Phi(K)\subset K$. Thus for $\psi\in \Phi(K)$,
\begin{align*}
\tau^{2-\eps'}\|P_M \psi(\tau)\|_{H^{\delta'}}+\tau^{3-3\eps'}\lvert
\<\psi(\tau),m_6\>\rvert&+\tau^{2-2\eps'-\delta'}\|\<y\>^{\delta'} P_M 
    \psi(\tau)\|_{L^2}\\
&\le \tau^{\eps-\eps'}+\tau^{3\eps-3\eps'}+\tau^{2\eps+\delta-2\eps'-\delta'}. 
\end{align*}
Let $\tau_1$ such that $\tau_0<\tau_1/2$ and 
\begin{equation*}
\(\frac{\tau_1}{2}\)^{\eps-\eps'}+\(\frac{\tau_1}{2}\)^{3\eps-3\eps'}
+\(\frac{\tau_1}{2}\)^{2\eps+\delta-2\eps'-\delta'}<\frac{r}{2}. 
\end{equation*}
From the two preceding inequalities, we get that for $\psi\in \Phi(K)$,
\begin{multline}
\label{cond_tau1}
\tau\ge\frac{\tau_1}{2}\Longrightarrow \\
\tau^{2-\eps'}\|P_M \psi(\tau)\|_{H^{\delta'}}+\tau^{3-3\eps'}\lvert
\<\psi(\tau),m_6\>\rvert+\tau^{2-2\eps'-\delta'}\|\<y\>^{\delta'} P_M 
    \psi(\tau)\|_{L^2}<\frac{r}{2}.
\end{multline}
Next, consider the set
$$F=\left\{ \Phi(w)\big|_{[\tau_0,\tau_1]},\; w\in K\right\}.$$
We will show that the assumptions of Theorem \ref{T:Simon} hold with 
$$ X=\Sigma^{\delta}, \quad B=\Sigma^{\delta'}, \quad Y=\Sigma^{\delta-2},$$
where we define (as $\delta<2$) $\Sigma^{\delta-2} =
H^{\delta-2}+\F(H^{\delta-2})$. Note that $0<\delta'<\delta$, so that
$X$ is compactly embedded in $B$.
The fact that $\Phi(K)\subset K$ shows that $F$ is a bounded subset
of $C\([\tau_0,\tau_1];\Sigma^{\delta}\)$. 
Furthermore if $\phi=\Phi(w)\in K$ then
$\phi=\phi_M+\phi_S$ where 
\begin{align*}
\partial_\tau \phi_S+iL\phi_S&=P_SR_p(w)+P_SZ_p(w)+Z_p(Q),\\
\partial_\tau \phi_M+iL\phi_M -P_MZ_p(\phi_M)&=P_MR_p(w)+P_MZ_p(w_S).
\end{align*}
Using that $\phi$ and $w$ are in $K$, we get that
$\partial_{\tau}\phi\in
C\([\tau_0,+\infty[;\Sigma^{\delta-2}\)$ and that
$\partial_{\tau}\phi_{\mid [\tau_0,\tau_1]}$ is uniformly bounded in
$\Sigma^{\delta-2}$ with a bound which is independent of
$\phi$. 
By Theorem~\ref{T:Simon}, $F$ has compact closure in
$C\left([\tau_0,\tau_1];\Sigma^{\delta}\right)$. As a consequence,
there exist $\tilde{\psi}_1$,\ldots $\tilde{\psi}_N$ such that  
\begin{equation}
\label{compactness_F}
\forall \tilde{\psi}\in F,\;\exists n\in\{1,\ldots N\},\quad
\sup_{\tau_0\le \tau\le \tau_1}\left\|\tilde{\psi}(\tau)-
\tilde{\psi}_n(\tau)\right\|_{\Sigma^{\delta'}}
<\frac{r}{6\tau_1^{\kappa}},  
\end{equation}
where $\kappa=\max\{2-\eps',3-3\eps',2-2\eps'-\delta'\}>0$.

Let $\chi\in C^{\infty}([\tau_0,+\infty[)$, supported in
$[\tau_0,\tau_1]$, such that $0\le \chi\le 1$, and $\chi=1$ on
$[\tau_0,\tau_1/2]$. For $1\le n\le N$, let
$\psi_n=\chi \tilde{\psi}_n$. We show that the $\psi_n$'s satisfy
\eqref{precompact}, which will conclude the proof of the lemma. Let
$\psi\in \Phi(K)$. By \eqref{compactness_F}, there exists
$n\in\{1,\ldots N\}$ such that 
$$
\left\|\chi\psi-\psi_n\right\|_{L^{\infty}\(\tau_0,+\infty,\Sigma^{\delta'}\)}
\le
\left\|\psi-\tilde{\psi}_n\right\|_{L^{\infty}\(\tau_0,\tau_1,\Sigma^{\delta'}\)}
<\frac{r}{A\tau_1^{\kappa}},$$
where $A$ is a large universal constant to be specified later. 
And thus, using that $\chi\psi$ is supported in $[\tau_0,\tau_1]$, 
$$ \forall \tau\ge\tau_{0},\quad
\left\|\chi\psi(\tau)-\psi_n(\tau)\right\|_{\Sigma^{\delta'}}<
\frac{r}{A\tau^{\kappa}}.$$ 
This implies, if $A$ is large enough,
\begin{multline*}
\tau^{2-\eps'}\|\chi P_M
\psi(\tau)-P_M\psi_n(\tau)\|_{H^{\delta'}}+\tau^{3-3\eps'}\lvert
\<\chi\psi(\tau)-\psi_n(\tau),m_6\>\rvert \\ 
+\tau^{2-2\eps'-\delta'}\|\<y\>^{\delta'} \left(\chi P_M
    \psi(\tau)-P_M\psi_n(\tau)\right)\|_{L^2}<\frac{r}{2}.
\end{multline*}
Furthermore, by \eqref{cond_tau1},
\begin{multline*}
\tau^{2-\eps'}\|(1-\chi)P_M
\psi(\tau)\|_{H^{\delta'}}+\tau^{3-3\eps'}\lvert
(1-\chi)\<\psi(\tau),m_6\>\rvert 
\\+\tau^{2-2\eps'-\delta'}\|(1-\chi)\<y\>^{\delta'} P_M
    \psi(\tau)\|_{L^2}<\frac{r}{2}.
\end{multline*}
Hence \eqref{precompact}. The proof is complete.
\end{proof}

\subsection{End of the proof}
The following proposition will allow us to use Schauder's Theorem in
order to prove Theorem \ref{theo:reduced}. 

\begin{proposition}\label{prop:continuityPhi}
   Let $0<\eps<\eps'<1/4$ and $1<\delta'<\delta<2$, and assume
  $2\eps+\delta<2\eps'+\delta'<2$.
The closed unit ball $K$ of $Y(\delta,\eps,\tau_0)$ is closed in
$Y(\delta',\eps',\tau_0)$. In addition, the map $\Phi:K\to K$ is
continuous for the topology of $Y(\delta',\eps',\tau_0)$. 
\end{proposition}
\begin{proof}
In view of Proposition~\ref{prop:stabilite}, we need only prove the
continuity. 
We start with an estimate of the difference of two parameters
$p,\tilde p$ defined from two different functions $w,\tilde w\in
Y(\delta',\eps',\tau_0)$.    
Recall that the existence of $p$ was proved in Proposition~\ref{P:p}
as a fixed point of the operator $\Psi(p)=\Psi_w(p)$, and that $w\in
Y(\delta',\eps',\tau_0)$ implies $\lVert
p\rVert_{3-3\eps',\tau_0}<\infty$. We have 
 \begin{align*}
   \left\lvert p(\tau)-\tilde p (\tau)\right\rvert &= \left\lvert
     \Psi_w(p)(\tau)-\Psi_{\tilde w}(\tilde p) (\tau)\right\rvert \\
   &\le \left\lvert \Psi_w(p)(\tau)-\Psi_{w}(\tilde
     p)(\tau)\right\rvert +\left\lvert \Psi_w(\tilde
     p)(\tau)-\Psi_{\tilde w}(\tilde p)(\tau)\right\rvert. 
  \end{align*}
By the contraction estimate \eqref{contractionp} on $\Psi_w$, we
get 
 \begin{equation*}
\lVert p-{\tilde p}\rVert_{3-3\eps',\tau_0}\le \frac{1}{1-\kappa}\lVert \Psi_w(\tilde
p)-\Psi_{\tilde w}(\tilde p)\rVert_{3-3\eps',\tau_0}. 
 \end{equation*}
Therefore, in view of the definition of $\Psi_w$, 
 \begin{align*}
\lVert p-{\tilde p}\rVert_{3-3\eps',\tau_0} &\lesssim \sum_{1\le j\le
  5}\left\lVert \( D_j(\tilde p)(w)-D_j(\tilde p)(\tilde
  w)\)\right\rVert_{3-3\eps',\tau_0}\\ &+\left\lVert
  \int_\tau^{+\infty} \( D_6(\tilde p)(w)-D_6(\tilde p)(\tilde
  w)\)\right\rVert_{3-3\eps',\tau_0}. 
 \end{align*}
By the definition \eqref{Dj} of $D_j(\tilde{p})$, one has 
 \begin{multline*}
 |D_j(\tilde p)(w)-D_j(\tilde p)(\tilde w)|\le \left| \< P_S(R_{\tilde
     p} (w)- R_{\tilde p} (\tilde w)),m_j\>\right|+ \left|\<
   P_SZ_{\tilde p}(w-\tilde w),m_j\>\right|\\ 
 \le \left| \< R_{NL,\tilde p} (w)- R_{NL,\tilde p} (\tilde
   w),m_j\>\right|+\left| \< R_{L,\tilde p} (w-\tilde w),m_j\>\right|
 +\left| \< P_SZ_{\tilde p}(w-\tilde w),m_j\>\right|. 
  \end{multline*}
In view of the explicit formulas for $R_{NL}$ and of the pointwise
estimate \eqref{RL_pointwise} on $R_L$, we infer, since $w,\tilde w\in
L^\infty_\tau H^1$: 
 \begin{multline*}
 |(D_j(\tilde p)(w)-D_j(\tilde p)(\tilde w))(\tau)|\lesssim \lVert
 w(\tau)-\tilde w(\tau)\rVert_{H^1} (\lVert w(\tau)\rVert_{H^1}+\lVert
 \tilde w(\tau) \rVert_{H^1})\\+\frac{1}{\tau^3} \lVert w(\tau)-\tilde
 w(\tau)\rVert_{H^1}+|\tilde{p}(\tau)|  \lVert w(\tau)-\tilde
 w(\tau)\rVert_{H^1}\lesssim \frac{1}{\tau^{4-2\eps'}} \lVert w-\tilde
 w\rVert_{\delta',\eps',\tau_0}. 
  \end{multline*}
In conclusion, 
\begin{equation}
\label{ppp}
\lVert p-{\tilde p}\rVert_{3-3\eps',\tau_0} \lesssim
\frac{1}{\tau_0^{\eps'}}\lVert w-\tilde
w\rVert_{\delta',\eps',\tau_0}. 
\end{equation}
Also, since 
\begin{equation*}
\Phi_1(w)(\tau)=-n_6\,\int_\tau^{+\infty} D_6(p)(w),
\end{equation*}
we get
\begin{equation*}
\lVert \Phi_1(w)-\Phi_1(\tilde w)\rVert_{\delta',\eps',\tau_0} \lesssim
\frac{1}{\tau_0^{\eps'}}\lVert w-\tilde
w\rVert_{\delta',\eps',\tau_0}. 
\end{equation*}
Therefore $\Phi_1$ is (Lipschitz-)continuous on
$Y(\delta',\eps',\tau_0)$. It remains to show the continuity of
$\Phi_2$.

Let $w\in K$  and $w_n\in K$ such that $w_n\to w$ in
$Y(\delta',\eps',\tau_0)$. Denote by
$\phi_n=\Phi_2(w_n)=P_M\Phi(w_n)$, and $\phi=\Phi_2(w)$. By
Lemma~\ref{L:compactness}, $\Phi(K)$ is 
relatively 
compact and there exists a 
subsequence of $\phi_n$  which converges in $Y(\delta',\eps',\tau_0)$
to some $\tilde{\phi}\in K$. It remains to show that
$\phi=\tilde{\phi}$. By \eqref{ppp} we have 
$$ \lim_{n\to \infty}\lVert p_n-p\rVert_{3-3\eps',\tau_0}=0. $$
By definition of $\Phi_2$, we have
$$ \partial_\tau \phi_n+iL\phi_n
-P_MZ_{p_n}(\phi_n)=P_MR_{p_n}(w_n)+P_M Z_{p_n}(P_Sw_n).$$ 
Letting $n$ tends to $\infty$, we get that $\tilde{\phi}$ satisfies
the following equation in the sense of distributions 
$$ \partial_\tau \tilde{\phi}+iL\tilde{\phi}
-P_MZ_{p}(\tilde{\phi})=P_MR_{p}(w)+P_M Z_{p}(P_Sw).$$ 
Using that $\phi$ is, by definition, solution to the same equation, we get 
$$ \partial_\tau (\tilde{\phi}-\phi)+iL(\tilde{\phi}-\phi)
-P_MZ_{p}(\tilde{\phi}-\phi)=0,$$ 
which implies, by Lemma \ref{lem:aprioriref}, that
$\phi=\tilde{\phi}$. The proof is complete. 
\end{proof}

\begin{proof}[Proof of Theorem \ref{theo:reduced}]
By Proposition~\ref{prop:continuityPhi}, $\Phi$ is a continuous map
from $K$ into itself. By Lemma \ref{L:compactness}, $\Phi(K)$ is
relatively compact in $Y(\delta',\eps',\tau_0)$. As $K$ is a convex
closed subset of $Y(\delta',\eps',\tau_0)$, we can apply Schauder's
Theorem (see \emph{e.g.} \cite[Corollary B.3]{TaylorIII}) which implies
that $\Phi$ has a fixed point  $w\in K$. By the definition of $K$ and
Proposition~\ref{P:p}, Theorem \ref{theo:reduced} follows.  
\end{proof}

\appendix

\section{A differential inequality}
\label{sec:diffineq}
\begin{lemma}
\label{L:diff}
Let $\mu>0$ and $m\in \N$. Let $(a_j)_{j=1\ldots
  m}$,  $(b_j)_{j=1\ldots m}$, be real constants, $a,b,\eta>0$,  
and $(\alpha_j)_{j=1\ldots m}$,
$(\beta_j)_{j=1\ldots m}$ be constants in $[0,1]$. Assume 
\begin{equation*}
\forall j\in \{1\ldots m\},\quad a_j+(b-a)\alpha_j >0,\quad b_j+(a-b)\beta_j>0.
\end{equation*} 
There exists $\tau_0$ such that for any $M>0$ and any
nonnegative continuous functions $z_1$ 
and $z_2$ on $[\tau_0,+\infty[$ such that
\begin{equation*}
  \sup_{\tau\ge\tau_0}\left|\tau^az_1(\tau)\right|+\sup_{\tau\ge\tau_0}|\tau^b
  z_2(\tau)| <\infty. 
\end{equation*}
and 
satisfying the following differential
inequality on $[\tau_0,+\infty[$: 
\begin{equation}
\label{diffeq}
\left\{
\begin{aligned}
 z_1(\tau) &\le \int_\tau^\infty\(
 \frac{M}{\sigma^{a+1+\eta}}+C\sum_{j=1}^m
\frac{z_1(\sigma)^{1-\alpha_j}z_2(\sigma)^{\alpha_j}}{\sigma^{a_j+1}}\)d\sigma,\\
z_2(\tau) &\le \int_\tau^\infty\(\frac{M}{\sigma^{b+1+\eta}}+
C\sum_{j=1}^m
\frac{z_1(\sigma)^{\beta_j}z_2(\sigma)^{1-\beta_j}}{\sigma^{b_j+1}}\)d\sigma,  
\end{aligned}
\right.
\end{equation}
we have
\begin{equation*}
  \sup_{\tau\ge\tau_0}\left|\tau^az_1(\tau)\right|+\sup_{\tau\ge\tau_0}|\tau^b
  z_2(\tau)| \le \mu M. 
\end{equation*}
\end{lemma}
\begin{proof}
Denote by
$$ Z_1(\tau)=\tau^a z_1(\tau),\quad Z_2(\tau)=\tau^b z_2(\tau).$$
Let $\widetilde a_j=(1-\alpha_j)a+\alpha_j b +a_j$, 
$\widetilde b_j=\beta_ja+(1-\beta_j) b+b_j$.
Using Young's inequality, $Z_1^{1-\theta}Z_2^{\theta}\le
(1-\theta)Z_1+\theta Z_2$, \eqref{diffeq} and H\"older inequality yield
$$ z_1(\tau)\le M\left\|
  \frac{1}{\sigma^{a+1+\eta}} \right\|_{L^{1}} +C\sum_{j=1}^m  
\Big( \|Z_1\|_{L^{\infty}}+\|Z_2\|_{L^{\infty}}\Big)
\left\|\frac{1}{\sigma^{\widetilde a_j+1}}\right\|_{L^{1}},$$ 
where the Lebesgue norms correspond to integration over
$[\tau,\infty[$. Similarly,
$$ z_2(t)\le M\left\| \frac{1}{\sigma^{b+1+\eta}}
\right\|_{L^{1}} +C\sum_{j=1}^m  
\Big(\|Z_1\|_{L^\infty}+\|Z_2\|_{L^\infty}\Big)
\left\|\frac{1}{\sigma^{\widetilde b_j+1}}\right\|_{L^{1}}.$$ 
For any $c>0$,
$$\left\|\frac{1}{\sigma^{c+1}}\right\|_{L^{1}([\tau,\infty[)}=
\frac{1}{c\tau^{c}},$$ 
and hence (with a constant $C$ depending only on the parameters
$\eta$,  $a$, $b$, $\widetilde a_j$, $\widetilde b_j$) 
\begin{align*}
  Z_1(\tau)&\le
  \frac{C}{\tau^{\eta}}M+
C\sum_{j=1}^m \frac{\|Z_1\|_{L^\infty([\tau_0,+\infty[)}+
\|Z_2\|_{L^\infty([\tau_0,+\infty[)}}{\tau^{\widetilde a_j-a}},\\
Z_2(\tau)&\le
\frac{C}{\tau^{\eta}}M
+C\sum_{j=1}^m \frac{\|Z_1\|_{L^\infty([\tau_0,+\infty[)}+
\|Z_2\|_{L^\infty([\tau_0,+\infty[)}}{\tau^{\widetilde b_j-b}}.
\end{align*}
By assumption, $\widetilde a_j-a$ and $\widetilde b_j-b$ are
positive. Taking the sup norm of the preceding inequalities and
using the triangle inequality, we get 
\begin{equation*}
 \begin{aligned}
\|Z_1\|_{L^\infty([\tau_0,+\infty[)}&\le
\frac{C}{\tau_0^{\eta}}M+
C\sum_{j=1}^m \frac{\|Z_1\|_{L^\infty([\tau_0,+\infty[)}+
\|Z_2\|_{L^\infty([\tau_0,+\infty[)}}{\tau_0^{\widetilde a_j-a}}\\
\|Z_2\|_{L^\infty([\tau_0,+\infty[)}&\le \frac{C}{\tau_0^{\eta}}
M+
C\sum_{j=1}^m \frac{\|Z_1\|_{L^\infty([\tau_0,+\infty[)}+
\|Z_2\|_{L^\infty([\tau_0,+\infty[)}}{\tau_0^{\widetilde b_j-b}}.
\end{aligned}
\end{equation*}
Taking  $\tau_0$ large, we obtain 
\begin{equation*}
 \begin{aligned}
\|Z_1\|_{L^\infty([\tau_0,+\infty[)}&\le
\frac{\mu}{4}M
+\frac{1}{4}\|Z_1\|_{L^\infty([\tau_0,+\infty[)}
+\frac{1}{4}\|Z_2\|_{L^\infty([\tau_0,+\infty[)},\\
\|Z_2\|_{L^\infty([\tau_0,+\infty[)}&\le
\frac{\mu}{4}M
+\frac{1}{4}\|Z_1\|_{L^\infty([\tau_0,+\infty[)}
+\frac{1}{4}\|Z_2\|_{L^\infty([\tau_0,+\infty[)}.
\end{aligned}
\end{equation*}
Summing up, we get the announced result.
\end{proof}

\section{Some interpolation inequalities}
\label{sec:interpol}

\begin{lemma}
  Let $d\ge 1$. There exists $C>0$ such that for all $f\in
  \Sch(\R^d)$, 
  \begin{align}
   \label{eq:inter1}
   \left\lVert \<x\> f\right\rVert_{L^2} &\le \left\lVert \<x\>^3
     f\right\rVert_{L^2}^{1/3}\left\lVert 
     f\right\rVert_{L^2}^{2/3},\\
\label{eq:inter2}
\left\lVert \<x\>^2 f\right\rVert_{L^2} &\le \left\lVert \<x\>^3
     f\right\rVert_{L^2}^{2/3}\left\lVert 
     f\right\rVert_{L^2}^{1/3},\\
\label{eq:inter3}
\left\lVert  f\right\rVert_{H^1} &\le \left\lVert 
     f\right\rVert_{H^3}^{1/3}\left\lVert 
     f\right\rVert_{L^2}^{2/3},\\
\label{eq:inter4}
\left\lVert  f\right\rVert_{H^2} &\le \left\lVert 
     f\right\rVert_{H^3}^{2/3}\left\lVert 
     f\right\rVert_{L^2}^{1/3},\\
  \label{eq:inter5}
   \left\lVert \<x\> \nabla f\right\rVert_{L^2} &\le C\left\lVert \<x\>^3
     f\right\rVert_{L^2}^{1/3}\left\lVert 
     f\right\rVert_{L^2}^{1/3}\left\lVert 
     f\right\rVert_{H^1}^{1/3},\\
  \label{eq:inter6}
   \left\lVert \<x\>^2 \nabla f\right\rVert_{L^2} &\le C\left\lVert \<x\>^3
     f\right\rVert_{L^2}^{2/3}\left\lVert 
     f\right\rVert_{H^1}^{1/3},\\
 \label{eq:inter7}
   \left\lVert \<x\> \nabla^2 f\right\rVert_{L^2} &\le C\left\lVert \<x\>^3
     f\right\rVert_{L^2}^{1/3}\left\lVert 
     f\right\rVert_{H^1}^{2/3},\\
  \label{eq:inter8}
   \left\lVert \<x\> \nabla^4 f\right\rVert_{L^2} &\le C\left\lVert \<x\>^5
     f\right\rVert_{L^2}^{1/5}\left\lVert 
     f\right\rVert_{H^5}^{4/5},\\
 \label{eq:inter9}
   \left\lVert \<x\>^4 \nabla f\right\rVert_{L^2} &\le C\left\lVert \<x\>^5
     f\right\rVert_{L^2}^{4/5}\left\lVert 
     f\right\rVert_{H^5}^{1/5}.
 \end{align}
\end{lemma}
\begin{proof}
  To prove \eqref{eq:inter1}, use H\"older's inequality:
  \begin{align*}
   \left\lVert \<x\> f\right\rVert_{L^2}^2 &=
   \int_{\R^d}\(\<x\>^6|f(x)|^2\)^{1/3} \(\lvert
   f(x)\rvert^2\)^{2/3}dx \\
&\le \left\lVert
     \(\<x\>^6|f(x)|^2\)^{1/3}\right\rVert_{L^3}\left\lVert \(\lvert
   f(x)\rvert^2\)^{2/3}\right\rVert_{L^{3/2}}.
  \end{align*}
Inequality \eqref{eq:inter2} follows the same way:
 \begin{align*}
   \left\lVert \<x\>^2 f\right\rVert_{L^2}^2 &=
   \int_{\R^d}\(\<x\>^6|f(x)|^2\)^{2/3} \(\lvert
   f(x)\rvert^2\)^{1/3}dx \\
&\le \left\lVert
     \(\<x\>^6|f(x)|^2\)^{2/3}\right\rVert_{L^{3/2}}\left\lVert \(\lvert
   f(x)\rvert^2\)^{1/3}\right\rVert_{L^{3}}.
  \end{align*}
Inequalities \eqref{eq:inter3} and \eqref{eq:inter4} then follow from
\eqref{eq:inter1} and \eqref{eq:inter2}, respectively, and
Plancherel formula. 

Integrating by parts, we have
\begin{align*}
 \left\lVert \<x\>\nabla f\right\rVert_{L^2}^2 &=
 -\int_{\R^d}\overline f(x) \nabla\cdot \(\<x\>^2 \nabla f(x)\)dx\\ 
&\lesssim \int_{\R^d} \lvert f(x) \rvert \<x\> \lvert \nabla
f(x)\rvert dx + \int_{\R^d} \lvert f(x) \rvert \<x\>^2 \lvert \Delta 
f(x)\rvert dx\\
&\lesssim \left\lVert \<x\>
     f\right\rVert_{L^2} \left\lVert  f\right\rVert_{H^1} +\left\lVert \<x\>^2
     f\right\rVert_{L^2} \left\lVert  f\right\rVert_{H^2} \\
&\lesssim \left\lVert \<x\>^3
     f\right\rVert_{L^2}^{1/3}\left\lVert 
     f\right\rVert_{L^2}^{4/3}\left\lVert 
     f\right\rVert_{H^1}^{1/3}  +\left\lVert \<x\>^3
     f\right\rVert_{L^2}^{2/3}\left\lVert 
     f\right\rVert_{L^2}^{2/3}\left\lVert 
     f\right\rVert_{H^1}^{2/3},
\end{align*}
where we have used \eqref{eq:inter1}--\eqref{eq:inter4}. Inequality
\eqref{eq:inter5} follows. 

Integration by parts also yields
\begin{align*}
  \left\lVert \<x\>^2 \nabla f\right\rVert_{L^2}^2 &= -\int_{\R^d} \overline
  f(x) \nabla\cdot \(\<x\>^4\nabla f(x)\)dx \\
&\lesssim  \left\lVert \<x\>^3
     f\right\rVert_{L^2}\left\lVert 
     f\right\rVert_{H^1} + \left\lVert \<x\>^3
     f\right\rVert_{L^2} \left\lVert \<x\> \nabla^2
     f\right\rVert_{L^2} \\
&\lesssim  \left\lVert \<x\>^3
     f\right\rVert_{L^2}\( \left\lVert 
     f\right\rVert_{H^3}^{1/3}\left\lVert 
     f\right\rVert_{L^2}^{2/3}+\left\lVert \<x\> \nabla^2
     f\right\rVert_{L^2}\). 
\end{align*}
On the other hand,
\begin{align}
 \left\lVert \<x\> \nabla^2
     f\right\rVert_{L^2}^2&= -\int_{\R^d} \nabla \overline f(x) \nabla
   \( \<x\>^2 \nabla^2 f(x)\)dx \notag \\
&\lesssim \left\lVert \<x\>\nabla f\right\rVert_{L^2} \left\lVert
   f\right\rVert_{H^2} + \left\lVert \<x\>^2 \nabla
   f\right\rVert_{L^2} \left\lVert 
   f\right\rVert_{H^3}\notag\\
&\lesssim \left\lVert 
   f\right\rVert_{H^3}\( \left\lVert \<x\>^3
     f\right\rVert_{L^2}^{1/3} \left\lVert  f\right\rVert_{L^2}^{2/3}
   + \left\lVert \<x\>^2 \nabla
   f\right\rVert_{L^2}\).\label{eq:dub}
\end{align}
We infer, for instance,
\begin{align*}
 \left\lVert \<x\>^2 \nabla f\right\rVert_{L^2}^2 &\lesssim
 \left\lVert \<x\>^3 f\right\rVert_{L^2}^{7/6}  \left\lVert 
     f\right\rVert_{L^2}^{1/3}\left\lVert 
   f\right\rVert_{H^3}^{1/2} + \left\lVert \<x\>^2 \nabla
   f\right\rVert_{L^2}^{1/2} 
\left\lVert \<x\>^3 f\right\rVert_{L^2}  \left\lVert 
   f\right\rVert_{H^3}^{1/2} \\
&\le C \left\lVert \<x\>^3 f\right\rVert_{L^2}^{7/6}  \left\lVert 
     f\right\rVert_{L^2}^{1/3}\left\lVert 
   f\right\rVert_{H^3}^{1/2} + \eps \left\lVert \<x\>^2 \nabla
   f\right\rVert_{L^2}^2  \\
&\quad+ C_\eps\(\left\lVert \<x\>^3
   f\right\rVert_{L^2}  \left\lVert  
   f\right\rVert_{H^3}^{1/2} \)^{4/3}, 
\end{align*}
where we have used Young's inequality, with $(4,4')=(4,4/3)$. Taking
$\eps<1$ yields \eqref{eq:inter6}, and \eqref{eq:inter7} then follows
from \eqref{eq:dub}.  

The proof of \eqref{eq:inter8} and \eqref{eq:inter9} is similar, and
we omit it. 
\end{proof}

\subsection*{Acknowledgments}

This work was supported by the French ANR projects:
  \emph{R.A.S}. (ANR-08-JCJC-0124-01); \emph{\'Etude qualitative des EDP}
(ANR-05-JCJC-0036); 
\emph{\'Equations de Gross-Pitaevski, d'Euler, et ph\'enom\`enes de
  concentration} 
(ANR-05-JCJC-51279); \emph{ONDNONLIN}; \emph{ControlFlux}. The authors
are grateful to the referees for their constructive comments.

\bibliographystyle{amsplain}
\bibliography{blowup}

\end{document}